\documentclass[10pt]{article}

\usepackage{charter} 
\usepackage{fullpage}
\usepackage{titlesec}
\usepackage{hyperref}
\usepackage{url}            
\usepackage{booktabs}       
\usepackage{amsfonts}       
\usepackage{amssymb}        
\usepackage{nicefrac}       
\usepackage{microtype}      
\usepackage{bm}				
\usepackage{mathrsfs}		
\usepackage{cite}			
\usepackage{graphicx}		
\usepackage{mathtools}		
\usepackage{algorithmic}	
\usepackage{algorithm}
\usepackage{amsthm}			
\usepackage{enumitem}		
\usepackage{multirow}       
\usepackage{tabularx}	    
\usepackage{hhline}
\usepackage{arydshln}		
\usepackage{wrapfig}
\usepackage{lipsum}

\renewcommand{\algorithmiccomment}[1]{\bgroup\hfill{$\triangleright$~#1}\egroup}
\DeclareMathOperator{\Tr}{Tr}
\DeclareMathOperator{\Cov}{Cov}

\newcolumntype{L}[1]{>{\raggedright\arraybackslash}p{#1}}
\newcolumntype{C}[1]{>{\centering\arraybackslash}p{#1}}
\newcolumntype{R}[1]{>{\raggedleft\arraybackslash}p{#1}}

\theoremstyle{plain} 
\newtheorem{proposition}{Proposition}
\newtheorem{definition}{Definition}
\newtheorem{theorem}{Theorem}
\newtheorem{lemma}{Lemma}

\newtheorem{assumption}{Assumption}

\def\defn{\,\coloneqq\,}
\def\argmin{\mathop{\mathsf{arg\,min}}} 

\def\lim{\mathop{\mathsf{lim}}} 
\def\min{\mathop{\mathsf{min}}}
\def\max{\mathop{\mathsf{max}}}

\def\zer{\mathsf{zer}}

\def\Dsf{\mathsf{\, d}}
\def\diag{\mathsf{diag}}

\def\ebm{{\bm{e}}}
\def\hbm{{\bm{h}}}

\def\xbm{{\bm{x}}}

\def\ybm{{\bm{y}}}

\def\Abm{{\bm{A}}}
\def\Dbm{{\bm{D}}}

\def\xbmast{{\bm{x}^\ast}}

\def\xbmtilde{{\widetilde{\bm{x}}}}

\def\Tsf{{\mathsf{T}}}

\def\Dsf{{\mathsf{D}}}

\def\Rsf{{\mathsf{R}}}
\def\Hsf{{\mathsf{H}}}

\def\Gsf{{\mathsf{G}}}
\def\Gsfhat{\widehat{\mathsf{G}}}
\def\Isf{{\mathsf{I}}}

\def\Hsf{{\mathsf{R}}}
\def\Usf{{\mathsf{U}}}
\def\Hsf{{\mathsf{H}}}

\def\R{\mathbb{R}}
\def\E{\mathbb{E}}
\def\N{\mathbb{N}}
\def\Z{\mathbb{Z}}

\def\Fcal{{\mathcal{F}}}

\def\Xcal{{\mathcal{X}}}

\def\proposed{$\textsc{Async-RED}$}

\begin{document}

\title{Async-RED: A Provably Convergent Asynchronous Block Parallel Stochastic Method using Deep Denoising Priors}

{\normalsize\author{Yu~Sun$^{\footnotesize 1}$, Jiaming~Liu$^{\footnotesize 3}$, Yiran~Sun$^{\footnotesize 3}$, Brendt~Wohlberg$^{\footnotesize 2}$,~and~Ulugbek~S.~Kamilov$^{\footnotesize 1, 3, \ast}$\\
\emph{\small $^{\footnotesize 1}$Department of Computer Science and Engineering,~Washington University in St.~Louis, MO 63130, USA}\\
\emph{\small $^{\footnotesize 2}$Los Alamos National Laboratory, Theoretical Division, Los Alamos, NM 87545 USA}\\
\emph{\small $^{\footnotesize 3}$Department of Electrical and Systems Engineering,~Washington University in St.~Louis, MO 63130, USA}\\
\small$^{\footnotesize *}$\emph{Email}: \texttt{kamilov@wustl.edu}
}}

\markboth{An Online Plug-and-Play Algorithm for Regularized Image Reconstruction}%
{Kamilov:  An Online Plug-and-Play Algorithm for Regularized Image Reconstruction}

\date{}
\maketitle 

\begin{abstract}
Regularization by denoising (RED) is a recently developed framework for solving inverse problems by integrating advanced denoisers as image priors.
Recent work has shown its state-of-the-art performance when combined with pre-trained \emph{deep denoisers}. However, current RED algorithms are inadequate for parallel processing on multicore systems. We address this issue by proposing a new \emph{asynchronous RED} (\proposed) algorithm that enables asynchronous parallel processing of data, making it significantly faster than its serial counterparts for large-scale inverse problems.
The computational complexity of \proposed~is further reduced by using a random subset of measurements at every iteration.
We present complete theoretical analysis of the algorithm by establishing its convergence under explicit assumptions on the data-fidelity and the denoiser. We validate \proposed~on image recovery using pre-trained deep denoisers as priors.
\end{abstract}

\section{Introduction}
\label{Sec:Intro}
Imaging inverse problems seek to recover an unknown image $\xbm\in\R^n$ from its noisy measurements $\ybm\in\R^m$.
Such problems arise in many fields, ranging from low-level computer vision to biomedical imaging.
Since many imaging inverse problems are ill-posed, it is common to regularize the solution by using prior information on the unknown image.
Widely-adopted image priors include total variation, low-rank penalties, and transform-domain sparsity~\cite{Rudin.etal1992, Figueiredo.Nowak2001, Figueiredo.Nowak2003, Hu.etal2012, Elad.Aharon2006}.

There has been considerable recent interest in \emph{plug-and-play priors (PnP)}~\cite{Venkatakrishnan.etal2013, Sreehari.etal2016} and \emph{regularization by denoising (RED)}~\cite{Romano.etal2017}, as frameworks for exploiting image denoisers as priors for image recovery. The popularity of deep learning has led to a wide adoption of \emph{deep denoisers} within PnP/RED, leading to their state-of-the-art performance in a variety of applications, including image restoration~\cite{Mataev.etal2019}, phase retrieval~\cite{Metzler.etal2018}, and tomographic imaging~\cite{Wu.etal2020}.
Their empirical success has also prompted a follow-up theoretical work clarifying the existence of explicit regularizers~\cite{Reehorst.Schniter2019}, providing new interpretations based on fixed-point projections~\cite{Cohen.etal2020}, and analyzing their coordinate/online variants~\cite{Sun.etal2019c,Wu.etal2020}.
Nonetheless, current PnP/RED algorithms are inherently \emph{serial}, which makes them suboptimal for large-scale inverse problems on multicore systems~(see Fig.~\ref{Fig:Scheme} for an illustration).

\begin{figure}[h]
\centering\includegraphics[width=0.99\linewidth]{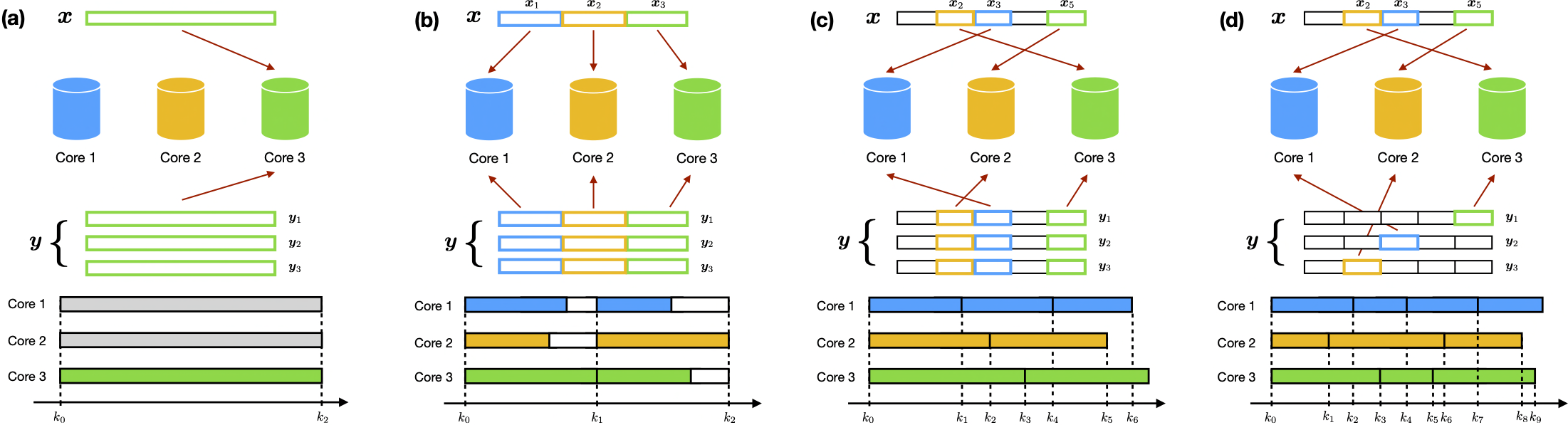}
\caption{Visual illustration of \emph{serial} and \emph{parallel} image recovery on a multicore system.
\textbf{(a)} Serial processing uses only one core of the system for every iteration. \textbf{(b)} Synchronous parallel processing has to wait for the slowest core to finish before starting the next iteration. \textbf{(c)} Asynchronous parallel processing can continuously iterate using all the cores without waiting. \textbf{(d)} Asynchronous parallel processing using the stochastic gradient leads to additional flexibility. \textbf{(a)}, \textbf{(b)}, and \textbf{(c)} use all the corresponding measurements at every iteration, while \textbf{(d)} uses only a small random subset at a time. \proposed~adopts the schemes shown in \textbf{(c)} and \textbf{(d)}.}
\label{Fig:Scheme}
\end{figure}

We address this gap by proposing a novel \emph{asynchronous RED} (\proposed) algorithm.
The algorithm decomposes the inference problem into a sequence of partial (block-coordinate) updates on $\xbm$ executed \emph{asynchronously} in parallel over a multicore system. \proposed~leads to a more efficient usage of available cores by avoiding synchronization of partial updates. \proposed~is also scalable in terms of the number of measurements, since it processes only a small random subset of $\ybm$ at every iteration.
We present two new theoretical results on the convergence of \proposed~based on a unified set of explicit assumptions on the data-fidelity and the denoiser.
Specifically, we establish its fixed-point convergence in the \emph{batch} setting and extend this analysis to the randomized \emph{minibatch} scenario.
Our results extend recent work on \emph{serial} block-coordinate RED~\cite{Sun.etal2019c} and are fully consistent with the traditional asynchronous parallel optimization methods~\cite{Lian.etal2015, Sun.etal2017}. We numerically validate \proposed~on image recovery from linear and noisy measurements using pre-trained deep denoisers as image priors.

\section{Background}
\label{Sec:Background}

\textbf{Inverse problems.}~Inverse problems are traditionally formulated as a composite optimization  problem
\begin{equation}
\label{Eq:Optimization}
\widehat{\xbm} = \argmin_{\xbm \in \R^{n}} g(\xbm) + h(\xbm),
\end{equation}
where $g$ is the data-fidelity term that ensures consistency of $\xbm$ with the measured data $\ybm$ and $h$ is the regularizer that infuses the prior knowledge on $\xbm$. For example, consider the smooth $\ell_2$-norm data-fidelity term $g(\xbm) = \|\ybm-\Abm\xbm\|_2^2$, which assumes a linear observation model $\ybm = \Abm\xbm + \ebm$, and the nonsmooth TV regularizer $h(\xbm) = \tau\|\Dbm\xbm\|_1$, where $\tau > 0$ is the regularization parameter and $\Dbm$ is the image gradient~\cite{Rudin.etal1992}.

\textbf{Regularization by denoising (RED).}~RED is a recent methodology for imaging inverse problems that seeks vectors $\xbmast \in \R^n$ satisfying
\begin{equation}
\label{Eq:FixedPoint}
\Gsf(\xbmast) = \nabla g(\xbmast) + \tau(\xbmast-\Dsf_\sigma(\xbmast)) = 0 \quad\Leftrightarrow\quad \xbmast\in\zer(\Gsf)\defn\{\xbm\in\R^{n} \,:\,\Gsf(\xbm)=0\}
\end{equation}
where $\nabla g$ denotes the gradient of the data-fidelity term and $\Dsf_\sigma:\R^n\rightarrow\R^n$ is an image denoiser parameterized by $\sigma>0$.
Under additional technical assumptions, the solutions $\xbmast \in \zer(\Gsf)$ can be associated with an explicit objective function of form \eqref{Eq:Optimization}.
Specifically, when $\Dsf_\sigma$ is locally homogeneous and has a symmetric Jacobian satisfying strong passivity~\cite{Romano.etal2017,Reehorst.Schniter2019},  $\Hsf(\xbm)$ corresponds to the gradient of a convex regularizer
\begin{equation}
\label{Eq:REDregularizer}
h(\xbm) = \frac{1}{2}\xbm^\Tsf(\xbm-\Dsf_\sigma(\xbm)).
\end{equation}

A simple strategy for computing $\xbmast \in \zer(\Gsf)$ is based on the following first-order fixed-point iteration
\begin{equation}
\label{Eq:REDupdate}
\xbm^t = \xbm^{t-1} - \gamma\Gsf(\xbm^{t-1}), \quad\text{with}\quad\Gsf\defn\nabla g + \tau(\Isf-\Dsf_\sigma), \quad\Gsf:\R^n\rightarrow\R^n,
\end{equation}
where $\gamma>0$ denotes the stepsize.
In this paper, we extend this first-order RED algorithm to design \proposed. Since many denoisers do not satisfy the assumptions necessary for having an explicit objective~\cite{Reehorst.Schniter2019}, our theoretical analysis considers a broader setting where $\Dsf_\sigma$ does not necessarily correspond to any explicit regularizer.
The benefit of our analysis is that it accommodates powerful deep denoisers (such as DnCNN~\cite{Zhang.etal2017}) that have been shown to achieve the state-of-the-art performance~\cite{Sun.etal2019c,Wu.etal2020,Cohen.etal2020}.

\textbf{\text{Plug-and-play priors (PnP) and other related work.}}~There are other lines of works that combine the iterative methods with advanced denoisers.
One closely-related framework is known as the \emph{deep mean-shift priors}~\cite{Bigdeli.etal2017}.
It develops an implicit regularizer whose gradient is specified by a denoising autoencoder.
Another well-known framework is PnP, which generalizes proximal methods by replacing the proximal map with an image denoiser~\cite{Venkatakrishnan.etal2013}.
Applications and theoretical analysis of PnP are widely studied in~\cite{Sreehari.etal2016, Zhang.etal2017a, Sun.etal2018b, Zhang.etal2019, Ahmad.etal2020, Wei.etal2020} and~\cite{Chan.etal2016, Meinhardt.etal2017, Buzzard.etal2017, Sun.etal2018a, Tirer.Giryes2019, Teodoro.etal2019, Ryu.etal2019, Xu.etal2020}, respectively.
In particular, \cite{Buzzard.etal2017} proposed a parallel extension of PnP called \emph{Consensus Equilibrium (CE)}, which enables synchronous parallel updates of $\xbm$. Note that while we developed \proposed~as a variant of RED, our framework and analysis can be also potentially applied to PnP/CE.
The plug-in strategy can be also applied to another family of algorithms known as \emph{approximate message passing (AMP)}~\cite{Metzler.etal2016a,Metzler.etal2016,Fletcher.etal2018}.
The AMP-based algorithms are known to be nearly-optimal for random measurement matrices, but are generally unstable for general $\Abm$~\cite{Rangan.etal2014, Rangan.etal2015}.

\textbf{\text{Asynchronous parallel optimization.}}%
~There are two main lines of work in asynchronous parallel optimization, the one involving the asynchrony  in coordinate updates~\cite{Liu.etal2013a, Peng.etal2015, Sun.etal2017, Hannah.etal2018, Hannah.etal2018b}, and the other focusing on the study of various asynchronous stochastic gradient methods~\cite{Recht.etal2011, Lian.etal2015, Liu.etal2018b, Zhou.etal2018, Lian.etal2018}.

Our work contributes to the area by developing a novel deep-regularized asynchronous parallel method with provable convergence guarantees.

\section{Asynchronous RED}
\label{Sec:Method}

\proposed~addresses the computational bottleneck by simultaneously considering the asynchronous partial updates of image $\xbm$ and the randomized usage of measurements $\ybm$.
In this section, we introduce the algorithmic details of our method.
We start with the basic batch formulation of \proposed~(\proposed-BG) followed by its minibatch variant (\proposed-SG).

\subsection{\proposed~using Batch Gradient}
When the gradient uses all the measurements $\ybm \in \R^m$, \proposed-BG~is the \emph{asynchronous} extension of the recent block-coordinate RED (BC-RED) algorithm~\cite{Sun.etal2019c}.
Consider the decomposition of the variable space $\R^n$ into $b \geq 1$ blocks
$$\xbm=(\xbm_{1},\cdots,\xbm_{b})\in\R^{n_1} \times \cdots \times \R^{n_b}=\R^n\quad\text{with}\quad n = n_1 + n_2 + \cdots + n_b,$$
For each $i \in \{1, \dots, b\}$, we introduce the operator $\Usf_i: \R^{n_i} \rightarrow \R^n$ that injects a vector in $\R^{n_i}$ into $\R^n$ and its transpose $\Usf_i^\Tsf$ that extracts the $i$th block from a vector in $\R^n$. This directly implies that
\begin{equation}
\Isf=\Usf_1\Usf_1^\Tsf + \cdots + \Usf_b\Usf_b^\Tsf \quad\text{and}\quad \|\xbm\|_2^2 = \|\xbm_1\|_2^2+ \cdots + \|\xbm_b\|_2^2 \quad\text{with}\quad\xbm_i=\Usf_{i}^\Tsf\xbm.
\end{equation}
In analogy to the RED operator $\Gsf$ in~\eqref{Eq:FixedPoint}, we define the block-coordinate operator $\Gsf_i$ as
\begin{equation}
\Gsf_i(\xbm) \defn \Usf_i\Usf_i^\Tsf\Gsf(\xbm), \quad\text{with}\quad\xbm \in \R^n \quad\text{and}\quad \Gsf_i:\R^n\rightarrow\R^n.
\end{equation}
Due to the asynchrony in the block updates, the iterate might be updated several times by different cores during a single update cycle of a core, which means that the evaluation of $\xbm^{k+1}$ relies on a \emph{stale} iterate $\xbmtilde^{k}$
\begin{equation}
\xbm^{k+1}\leftarrow\xbm^{k}-\gamma\Gsf_{i_k}(\xbmtilde^{k}), \quad\text{with}\quad \xbmtilde^{k} = \xbm^{k} + \sum_{s = k-\Delta_k}^{k-1}(\xbm^s - \xbm^{s+1}),\quad\Delta_k \leq \lambda.
\end{equation}
Here, we assume that the stale iterate $\xbmtilde^{k}$ exits as a state of $\xbm$ in the shared memory, and the delay between them is bounded by a finite number $\lambda\in\Z_+$. 
These two assumptions are often referred to as the \emph{consistent read}~\cite{Recht.etal2011} and the \emph{bounded delay}~\cite{Liu.etal2015b} in the traditional asynchronous block coordinate optimization. 
Although we implement the consistent read in \proposed, the algorithm never imposes a global lock on $\xbm^k$. We refer to Supplement~\ref{Sup:Memory} for the related discussion.

We now introduce the first variant, \proposed-BG.
\begin{algorithm}[H]
\caption{\proposed-BG}
\label{Alg:AsyncREDbg}
\begin{algorithmic}[1]
\STATE \textbf{input: } $\xbm^0 \in \R^n$, $\gamma > 0$, $\tau > 0$.
\STATE \textbf{setup: } A multicore system with one shared memory storing $\xbm$ and global iteration $k$.
\FOR{$\textbf{global}\;k=1,2,3,\dots$}
\STATE $\xbmtilde^{k} \leftarrow \mathsf{read}(\xbm)$ 
\STATE $\Gsf_{i_k}(\xbmtilde^{k}) \leftarrow \Usf_{i_k}\Usf_{i_k}^\Tsf\Gsf(\xbmtilde^k)$
\quad with random $i_k\in\{1,\dots,b\}$ \COMMENT{Block Operation}
\STATE $\xbm^{k} \leftarrow \mathsf{read}(\xbm)$ 
\STATE $\xbm^{k+1} \leftarrow \xbm^{k} - \gamma\Gsf_{i_k}(\xbmtilde^{k})$
\STATE update $\xbm$ in the shared memory using $\xbm^{k+1}$
\ENDFOR
\end{algorithmic}
\end{algorithm}%
When the algorithm is run on a single core system without parallelization (that is to say $\xbmtilde^k=\xbm^k$), it reduces to the normal BC-RED algorithm.
Hence, our analysis is also applicable to BC-RED.

We specifically consider the \emph{random} block selection strategy in \proposed-BG, namely that every block index $i_k$ is selected as an i.i.d random variable uniformly distributed over $\{1,\dots,b\}$.
Such a strategy is commonly adopted for simplifying the convergence analysis.
Nevertheless, our method and analysis can be generalized to the scenario where $i_k$ follows some arbitrary probability $P(i_k=i)=p_i$ specified by the user.

Compared with serial RED algorithms, \proposed-BG enjoys considerable scalability by dividing the computation of the full operator $\Gsf$ into $b$ parallel evaluation of $\Gsf_i$ distributed across all cores.
Thus, without any modification to the algorithmic design, one can easily improve the performance of the algorithm by simply integrating more cores into the system.
In Section~\ref{Sec:Experiments}, we experimentally demonstrate the significant speed-up and scale-up in solving the context of image recovery.


\subsection{Async-RED using Stochastic Gradient}

The scale of measurements is another important factor influencing the computational complexity in the large-scale inference tasks.
\proposed-SG improves the applicability of \proposed~to these cases by further considering the decomposition of the measurement space $\R^m$ into $\ell\geq1$ blocks
$$\ybm=(\ybm_{1},\cdots,\ybm_{\ell})\in\R^{m_1} \times \cdots \times \R^{m_\ell}=\R^m\quad\text{with}\quad m = m_1 + m_2 + \cdots + m_\ell.$$
Hence, \proposed-SG~considers the following data-fidelity $g$ and its gradient $\nabla g$
\begin{equation}
\label{Eq:StochData}
g(\xbm)=\frac{1}{\ell}\sum_{j=1}^{\ell}g_j(\xbm)\quad\Rightarrow\quad \nabla g(\xbm)=\frac{1}{\ell}\sum_{j=1}^{\ell}\nabla g_j(\xbm),
\end{equation}
where each $g_j$ is evaluated on the subset $\ybm_j\in\R^{m_j}$ of the full $\ybm$.
From~\eqref{Eq:StochData}, we know that the computation of $\nabla g(\xbm)$ is proportional to the total number $\ell$.
To reduce the per-iteration cost, we follow the idea of stochastic optimization to approximate the batch gradient by using the stochastic gradient that relies on a minibatch of $w\ll\ell$ measurements
\begin{equation}
\label{Eq:StochGradient}
\widehat{\nabla} g(\xbm)=\frac{1}{w}\sum_{s=1}^{w}\nabla g_{j_s}(\xbm),
\end{equation}
where $j_s$ is picked from the set $\{1,\dots,\ell\}$ as i.i.d uniform random variable. Based on the minibatch gradient, we define the block stochastic operator $\Gsfhat_i:\R^n\rightarrow\R^n$ as
\begin{equation}
\Gsfhat_i \defn \Usf_i\Usf_i^\Tsf\Gsfhat(\xbm), \quad\text{with}\quad\Gsfhat\defn\widehat{\nabla} g(\xbm) + \tau(\xbm-\Dsf_\sigma(\xbm)), \quad\Gsfhat:\R^n\rightarrow\R^n.
\end{equation}
Note that the computation of $\Gsfhat_i$ is now dependent on the minibatch size $w$ that is adjustable to cope with the computation resources at hand. \proposed-SG is summarized in Algorithm~\ref{Alg:AsyncREDsg}.
\begin{algorithm}[H]
\caption{\proposed-SG}
\label{Alg:AsyncREDsg}
\begin{algorithmic}[1]
\STATE \textbf{input: } $\xbm^0 \in \R^n$, $\gamma > 0$, $\tau > 0$.
\STATE \textbf{setup: } A multicore system with one shared memory storing $\xbm$ and global iteration $k$.
\FOR{$\textbf{global}\;k=1,2,3,\dots$}
\STATE $\xbmtilde^{k} \leftarrow \mathsf{read}(\xbm)$ 
\STATE $\Gsfhat(\xbmtilde^{k}) \leftarrow \mathsf{minibatchG}(\xbmtilde^{k}, w)$\quad with random $j_w\in\{1,\dots,\ell\}$ \COMMENT{Minibatch Gradient}
\STATE $\Gsfhat_{i_k}(\xbmtilde^{k}) \leftarrow \Usf_{i_k}\Usf_{i_k}^\Tsf\Gsfhat(\xbmtilde^k)$ \quad with random $i_k\in\{1,\dots,b\}$ \COMMENT{Block Operation}
\STATE $\xbm^{k} \leftarrow \mathsf{read}(\xbm)$ 
\STATE $\xbm^{k+1} \leftarrow \xbm^{k} - \gamma\Gsfhat_{i_k}(\xbmtilde^{k})$
\STATE update $\xbm$ in the shared memory using $\xbm^{k+1}$
\ENDFOR
\end{algorithmic}
\end{algorithm}%
\vspace{-10pt}
We clarify the difference between \proposed-SG and \proposed-BG via a specific example.
Consider the least-squares $g$ with a block-friendly operator $\Abm$ and a block-efficient denoiser $\Dsf_\sigma$.
We can write the update of \proposed-BG regarding a single iteration as
\begin{equation}
\Gsf_i(\xbmtilde) = \Abm_{i}^\Tsf (\Abm_i\xbmtilde-\ybm_i) + \tau (\xbmtilde_i - \Dsf(\xbmtilde_i)),
\end{equation}
where $\xbmtilde$ is the delayed iterate for $\xbm$, and $\Abm_{i} \in \R^{m \times n_i}$ is a submatrix of $\Abm$ consisting of columns corresponding to the $i$th blocks. 
Although the per-iteration complexity is reduced by roughly $b=n/n_i$ times by working with $\Abm_i$ instead of $\Abm$, \proposed-BG still needs to work with all the measurements $\ybm_i$ related to the $i$th block at every iteration.
Consider the corresponding update of \proposed-SG with one measurement used at a time
\begin{equation}
\label{Eq:SGupdate}
\Gsfhat_i(\xbmtilde) = \Abm_{ji}^\Tsf (\Abm_{ji}\xbmtilde-\ybm_{ji}) + \tau (\xbmtilde_i - \Dsf(\xbmtilde_i)),
\end{equation}
where $\ybm_{ji}$ denotes the $j$th measurement of $\xbm_i$,
and $\Abm_{ji} \in \R^{m_j \times n_i}$ is the submatrix crossed by the rows and columns corresponding to the $j$th measurement and the $i$th blocks.
This indicates that the reduction of the per-iteration complexity from \proposed-BG to \proposed-SG can be up to $\ell=m/m_j$ times.
In the practice, it is common to use $w>1$ measurements at a time to optimize the total runtime.
Note that if $\Usf=\Usf^\Tsf=\Isf$, \proposed-SG becomes the asynchronous stochastic RED algorithm.
In the next section, we will present a complete analysis of \proposed~and theoretically discuss its connection to the related algorithms.

\section{Convergence Analysis of \proposed}
\label{Sec:Theory}

The proposed analysis is based on the following explicit assumptions. Note that these assumptions serve as sufficient conditions for the convergence.
\begin{assumption}
\label{As:BoundedDelay}
We assume bounded maximal delay $\lambda < \infty$. Hence, during any update cycle of an agent, the estimate $\xbm$ in the shared memory is updated at most $\lambda\in\Z_+$ times by other cores.
\end{assumption}
The value of $\lambda$ is often dependent on the number of cores involved
in the computation~\cite{Wright2015}.
If every core takes a similar amount of time to compute its update, $\lambda$ is expected to be a multiple of the number of cores.
Related work has investigated the convergence with unbounded maximal delays in the context of traditional optimization~\cite{Hannah.etal2018, Peng.etal2019, Zhou.etal2018}.
\begin{assumption}
\label{As:NonemptySet}
The operator $\Gsf$ is such that $\zer(\Gsf) \neq \varnothing$, and the distance of the initial $\xbm^0 \in \R^n$ to any element in $\zer(\Gsf)$ is bounded, that is
$\|\xbm^0-\xbmast\| \leq R_0$ for all $ \xbmast \in \zer(\Gsf)$ with $R_0<\infty$.
\end{assumption}
This assumption ensures the existence of a solution for the RED problem and is related to the existence of minimizers in traditional coordinate minimization~\cite{Nesterov2012, Beck.Tetruashvili2013}


\begin{assumption}
\label{As:DataFitConvexity}
\textbf{(a)} Every component function $g_i$ is convex differentiable and has a Lipschitz continuous gradient of constant $L_i>0$. \textbf{(b)} At every update, the stochastic gradient is unbiased estimator of $\nabla g$ that has a bounded variance:
$$\E\left[\widehat{\nabla} g(\xbm)\right]=g(\xbm),\quad\E\left[\|\widehat{\nabla} g(\xbm)-\nabla  g(\xbm)\|^2\right]\leq\frac{\nu^2}{w},\quad \xbm \in \R^n,\quad\nu>0.$$
\end{assumption}
The first part of the assumption implies that $g$ is also convex and has Lipschitz continuous gradient with constant $L = \max\{L_1, \dots, L_\ell\}$. The second part is a standard assumption on the unbiasedness and variance of the stochastic gradient~\cite{Lian.etal2015,Ghadimi.Lan2016}. Our final assumption is related to the deep denoiser used in \proposed.
\begin{assumption}
\label{As:NonexpansiveDen}
The denoiser $\Dsf_\sigma$ is a nonexpansive operator $\|\Dsf_\sigma(\xbm)-\Dsf_\sigma(\ybm)\|\leq\|\xbm-\ybm\|$.
\end{assumption}
Compared with the conditions stated in Section~\ref{Sec:Background} (namely, that it is locally homogeneous with a symmetric Jacobian), our requirement on the denoiser is milder.
One can train a nonexpansive $\Dsf_\sigma$ by constraining the Lipschitz constant of $\Dsf_\sigma$ via the spectral normalization, which is an active area of research in deep learning~\cite{Miyato.etal2018,Sedghi2018,Anil.etal2019}.

We can now state the theorems on \proposed.

\begin{theorem}
\label{Theo:ConvergenceBG}
Let Assumptions~\ref{As:BoundedDelay}-\ref{As:NonexpansiveDen} hold true. Run \proposed-BG for $t>0$ iterations with uniform i.i.d block selection using a fixed step-size $\gamma\in(0,1/((1+2\lambda)(L+2\tau))]$. Then, the iterates of the algorithm satisfy
\begin{equation}
\label{Eq:ConvergenceBG}
\min_{0\leq k\leq t-1}\E\left[\|\Gsf(\xbm^k)\|^2\right]\leq\left[\frac{D}{b}+2\right]\frac{(L+2\tau)b}{\gamma t}R_0^2.
\end{equation}
where $D=2\lambda^2/(1+\lambda)^2$ is a constant.
\end{theorem}
Theorem~\ref{Theo:ConvergenceBG} establishes the convergence of \proposed-BG to the fixed-point set $\zer(\Gsf)$ at the rate of $O(1/t)$.
Our result is consistent with the existing results in the literature. In particular, when the algorithm adopts serial block updates, that is $\lambda=0$ and $\xbmtilde^k=\xbm^k$, the recovered convergence is nearly the same as BC-RED~\cite{Sun.etal2019c} scaled by some constant.
On the other hand, our convergence rate $O(1/t)$ is also consistent with the rate proved for the asynchronous block coordinate descent in nonconvex optimization~\cite{Sun.etal2017}.
\begin{theorem}
\label{Theo:ConvergenceSG}
Let Assumptions~\ref{As:BoundedDelay}-\ref{As:NonexpansiveDen} hold true. Run \proposed-SG for $t>0$ iterations with uniform i.i.d selections of blocks and measurements using a fixed step-size $\gamma\in(0,1/((1+2\lambda)(L+2\tau))]$. Then, the iterates of the algorithm satisfy
\begin{equation}
\label{Eq:ConvergenceSG}
\min_{0\leq k\leq t-1}\E\left[\|\Gsf(\xbm^k)\|^2\right]\leq\left[\frac{D}{b}+2\right]\frac{(L+2\tau)b}{\gamma t}R_0^2+\left[\frac{2D}{b}+2\right]\frac{\gamma}{w}C
\end{equation}
where $C=(L+2\tau)(1+\lambda)\nu^2$ and $D=2\lambda^2/(1+\lambda)^2$ are constants.
\end{theorem}
Theorem~\ref{Theo:ConvergenceSG} states that \proposed-SG approximates the solution obtained by \proposed-BG up to a finite error that decreases for larger values of the minibatch size $w$. This relationship is consistent with the recent theoretical results on the online PnP and RED algorithms~\cite{Sun.etal2018a,Wu.etal2020}.
In practice, the selection of $w$ must balance the actual memory capacity of the system and the desired runtime for obtaining a reasonable solution.
Our numerical evaluation in Section~\ref{Sec:Experiments} demonstrates the excellent approximation of \proposed-SG to the batch-gradient solution by using a small subset of data.

\begin{figure}[t]
\centering\includegraphics[width=0.99\linewidth]{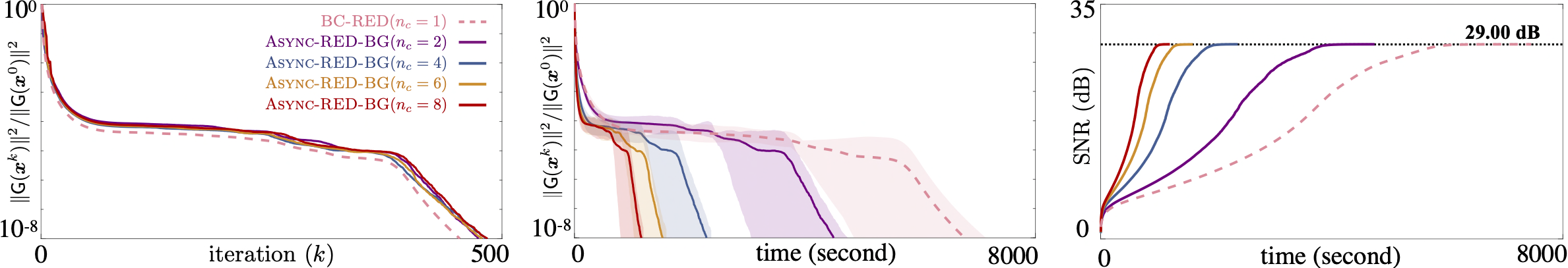}
\caption{Convergence of \proposed-BG for different numbers of accessible cores $n_c\in\{2,4,6,8\}$.
The left figure plots the average normalized distance to $\zer(\Gsf)$ against the iteration number; the middle and right figures plot these values, as well as SNR, plotted against the actual runtime in seconds.
The shaded areas represent the range of values attained over the test images.
}
\label{Fig:ConvergenceBG}
\end{figure}

By carefully choosing the stepsize $\gamma$, we can state the following remark on Theorem~\ref{Theo:ConvergenceSG}. \\
\textbf{Remark 1.} Set the stepsize to be $\gamma=1/\sqrt{wt}$. If the maximal delay satisfies $\lambda\leq(1/2)[\sqrt{wt}/(L+2\tau)-1]$, then after $t>0$ iterations we have
\begin{equation}
\min_{0\leq k\leq t-1}\E\left[\|\Gsf(\xbm^k)\|^2\right]
\leq\left[\frac{D}{b}+2\right]\frac{(L+2\tau)b}{\sqrt{wt}}R_0^2+\left[\frac{2D}{b}+2\right]\frac{C}{\sqrt{wt}}.
\end{equation}
This establishes the fixed-point convergence to the set $\zer(\Gsf)$ at the rate of $O(1/\sqrt{wt})$ under specific conditions. If we treat entire $\xbm$ as a block, namely that $\Usf=\Usf^\Tsf=\Isf$ and $b=1$, \proposed-SG then becomes the asynchronous stochastic RED algorithm.
Hence, the proposed remark immediately holds true for the later.
Note that our convergence rate $O(1/\sqrt{wt})$ is consistent with the rate proved for the serial~\cite{Nemirovski.etal2009} and parallel~\cite{Dekel.etal2010, Lian.etal2015} stochastic gradient methods.

All the proofs are presented in the supplement.
Our analysis never assumes the existence of an explicit regularizer, and hence applicable to advanced denoisers that are not associated with any regularizer.

\begin{figure}[t]
\centering\includegraphics[width=0.99\linewidth]{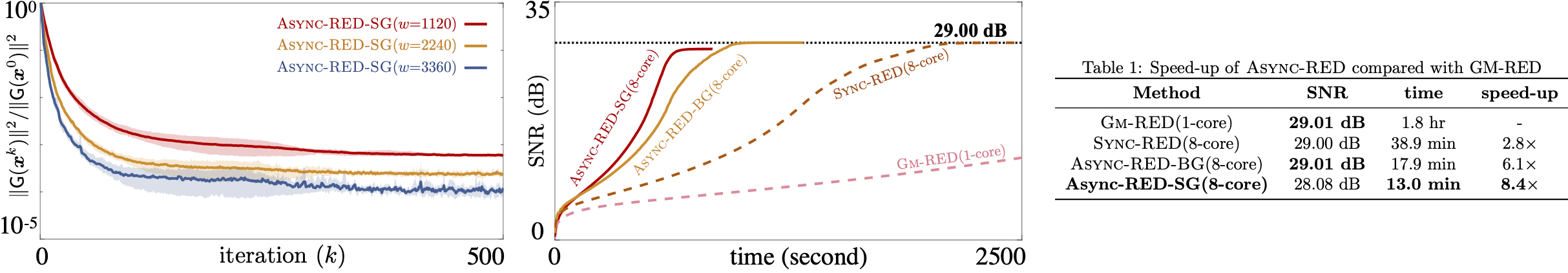}
\caption{\emph{\textbf{Left:}} Evolution of the convergence accuracy of \proposed-SG as the minibatch size $w$ increases.
The average distance is plotted against the number of iterations with the shaded areas representing the range of values attained over the test images.
\emph{\textbf{Middle \& Right:}} Comparison of convergence speed between \proposed-BG/SG~and other baselines. The right table summarizes the total runtime and the speed-up compared with \textsc{Gm-RED} for all algorithms.
}
\label{Fig:ConvergenceSG}
\end{figure}

\section{Numerical Validation}
\label{Sec:Experiments}

We now present a numerical validation of \proposed.
Our goals are first to validate the proposed theorems in Section~\ref{Sec:Theory} and then to demonstrate the effectiveness and the efficiency of our algorithm on the large-scale problem.
We consider two image recovery tasks that have the form $\ybm=\Abm\xbm+\ebm$, where the measurement matrix $\Abm$ corresponds to either the random matrix in \emph{compressive sensing (CS)} or the Radon transform in \emph{computed tomography (CT)}, and the noise $\ebm$ is assumed to be additive white Gaussian (AWGN).
In particular, the random matrix is implemented with the block-diagonal structure $\Abm = \diag([\Abm_i,...,\Abm_b])$ for fast validation, while the Radon transform is used as its full matrix form to demonstrate the effectiveness of \proposed~for overcoming the computation bottleneck.
Our deep neural net prior adapts the DnCNN architecture~\cite{Zhang.etal2017}.
We used the signal-to-noise ratio (dB) to quantify the quality of the reconstructed images.
For each experiments, we selected the denoiser that achieves the best SNR performance from the ones corresponding to five noise levels $\sigma\in \{5,10,15,20,25\}$. Supplement~\ref{Sup:Experiments} provides additional technical details.

\subsection{Convergence Behavior}
We validate our theorems on the CS task with $6$ test images selected from the $\emph{Set 12}$ dataset~\cite{Zhang.etal2017}.
Each test image is rescaled to the size of $240\times240$ pixels (see Fig.~\ref{Fig:TestImages} in the supplement for the visualization).
The block-diagonal matrix $\Abm$ is set to consist of $9$ submatrices, corresponding to a $3\times3$ grid of blocks with the size of $80\times80$ pixels in every image.
The elements in $\Abm$ are i.i.d zero-mean Gaussian random variables of variance of $1/m$, and the compression ratio is set to be $m/n=0.7$, which indicates that the total number of measurements is $4480$ for each block.
We obtain the measurements by multiplying $\Abm$ with each vectorized image and adding additional noise corresponding to the input SNR of $30$ dB.
Finally, we use the normalized distance $\|\Gsf(\xbm^k)\|_2^2/\|\Gsf(\xbm^0)\|_2^2$ to quantify the fixed-point convergence, with $b$ block updates grouped as one iteration.
The distance is expected to approach zero as the algorithm converges to a fixed point.

Theorem~\ref{Theo:ConvergenceBG} establishes the convergence of \proposed-BG to the fixed point set $\zer(\Gsf)$.
This is illustrated in Fig.~\ref{Fig:ConvergenceBG} for four different numbers of accessible cores $n_c\in\{2,4,6,8\}$.
In the left figure, the average normalized distance is plotted against the iteration number, while the middle and right figures plot the corresponding distance and SNR values against the actual runtime in seconds.
The shaded areas representing the range of values attained across all test images.
We also plot the results of serial BC-RED using the dashed line as reference.
\proposed-BG is implemented to be run asynchronously on multiple cores, while BC-RED can only use one core to perform the computation.
The left figure highlights the fixed-point convergence of \proposed-BG in iteration for different $n_c$, with all variants agreeing with the serial BC-RED.
Since \proposed-BG uses more cores, the middle and right figures demonstrate the significantly faster in-time convergence of \proposed-BG than BC-RED to the same SNR value.
Specifically, BC-RED takes $1.8$ hours to achieve $29.00$ dB, while \proposed-BG ($n_c=8$) takes only $17.9$ minutes to obtain the same value, corresponding to a $6\times$ improvement in computation time.

Theorem~\ref{Theo:ConvergenceSG} establishes the convergence of \proposed-SG to $\zer(\Gsf)$ up to some error term, which is inversely proportional to the minibatch size $w$.
This is illustrated in Fig.~\ref{Fig:ConvergenceSG} (left) for three different minibatch sizes $w\in\{1120,2240,3360\}$.
As before, we plotted the average distance against the iteration number with the shading area representing the variance.
Note that the log-scale of y-axis highlights the change for smaller values.
Fig.~\ref{Fig:ConvergenceSG} demonstrates the improved convergence of \proposed-SG to $\zer(\Gsf)$ for larger $w$, which is consistent with our theoretical analysis.
Fig.~\ref{Fig:ConvergenceSG} (middle) compares the convergence speed between \proposed-BG/SG, gradient-method RED (\textsc{Gm-RED}), and synchronous parallel RED (\textsc{Sync-RED}).
For \proposed-SG, we use $w=1120$.
In particular, \proposed-SG takes fewer total runtime (from 17.9 min to 13.0 min) to obtain the similar result ($29.01$ dB and $28.03$ dB) and achieves $8.4\times$ speedup compared with $\textsc{Gm-RED}$. The table in  Fig.~\ref{Fig:ConvergenceSG} summarizes the detailed results.

\begin{figure}[t]
\centering\includegraphics[width=0.99\linewidth]{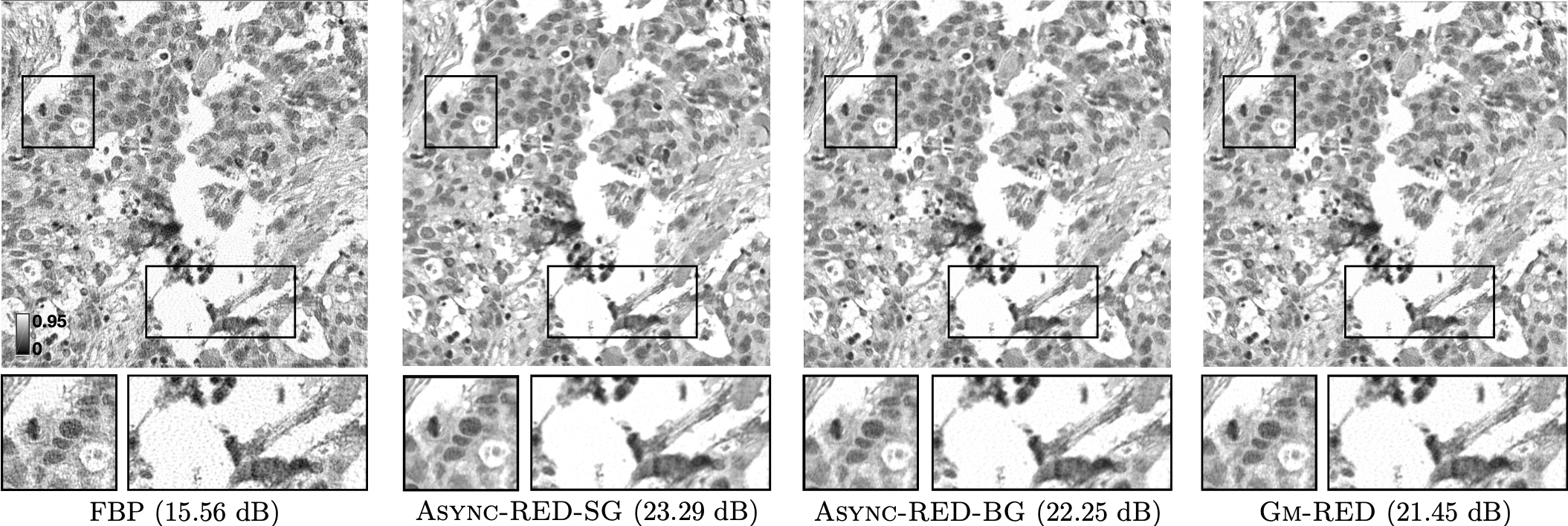}
\caption{CT reconstruction with a time budget of $1$ hour by \proposed-BG/SG and \textsc{Gm-RED}. The colormap is adjusted for the best visual quality.}
\label{Fig:VisualCT}
\end{figure}

\subsection{Effectiveness for Computational Imaging}
We additionally demonstrate the effectiveness of our algorithm~ by reconstructing a $800\times800$ CT image from its 180 projections.
For block parallel updates, the image is decomposed into $16$ blocks, each having the size of $200\times200$ pixels.
The Radon matrix used in the experiment corresponds to $180$ angles with $1131$ detectors, and the noise level is set to $70$ dB.
We refer to Supplement~\ref{Sec:ExtraValidations} for additional technical details.
Fig.~\ref{Fig:VisualCT} shows the visual illustration of the reconstructed images by \proposed-BG/SG and \textsc{Gm-RED}.
Each algorithm starts from the filtered back-projection (FBP) of the measurements and runs for $1$ hour.
Here, \proposed-SG randomly uses one-third of the total measurements at every iteration.
Given the same amount of time, \proposed-BG/SG successfully mitigates the noise-artifacts, while the result of \textsc{Gm-RED} is still noisy.
In particular, the per-iteration time cost of \proposed-BG/SG and \textsc{Gm-RED} is $5.23$, $3.21$, and $19.19$ seconds, respectively. This experiment clearly illustrates the fast processing speed of the asynchronous procedure.

\section{Conclusion}
\label{Sec:Conclusion}

Asynchronous parallel methods have gained increasing importance in optimization for solving large-scale imaging inverse problems.
We have introduced \proposed~as an extension of the recent RED framework and theoretically analyze its convergence in batch and stochastic settings.
We have validated its convergence guarantees and demonstrated its effectiveness in CT image reconstruction.
Future work will investigate theoretical limits of \proposed~in the unbounded maximal delay setting and explore its applicability to various inference problems in other data-intensive fields.

\bibliographystyle{IEEEtran}

\newpage
\appendix
\section*{Supplementary Material}

Our unified analysis of \proposed~is based on the monotone operator theory~\cite{Ryu.Boyd2016}.
In Supplement~\ref{Sup:Memory}, we first clarify our setting for the access of the shared memory.
In Supplement~\ref{Sup:Proof}, we present the proof of Theorem~\ref{Theo:ConvergenceBG} and Theorem~\ref{Theo:ConvergenceSG}, proving the fixed-point convergence of \proposed~to $\zer(\Gsf)$ in both batch and stochastic settings.
In Supplement~\ref{Sup:Review}, we provide a brief review of the related knowledge on monotone operators.
In Supplement~\ref{Sup:Experiments}, we include additional technical details and experiments omitted from the main paper due to space.

\section{Memory Access without Global Lock}
\label{Sup:Memory}

In the setting of \proposed, multiple cores may simultaneously read and update the blocks $\xbm_i$ in shared memory.
We coordinate the memory access of different cores by imposing certain \emph{local} locks.
For example, consider one work cycle of core $c_i$ for updating the block $\xbm_i$. First, a local \emph{read} lock is imposed to $\xbm_i$ such that only read operations (by $c_i$ or others) can be performed on $\xbm_i$.
If, at the same time, other cores want to write $\xbm_i$, then they have to wait until the read lock is released by the last one who finishes reading the block.
However, if they want to write other blocks, their operations will not be blocked.
Secondly, core $c_i$ evaluates the RED update on $\xbm_i$, while other cores continuously update $\xbm$.
Here, we assume that the number of updates by cores other than $c_i$ is bounded by some positive integer, which is exactly what Assumption~\ref{As:BoundedDelay} refers to.
After the evaluation finishes, core $c_i$ imposes a local \emph{write} lock, which prevents both read and write by other cores, on $\xbm_i$ and write the block with the computed update.
Similarly, other cores have to wait until the lock is released before operating on $\xbm_i$.
Finally, when the update finishes, the local lock will be released and core $c_i$ will restart a new cycle.
Note that $\xbm$ is never locked \emph{globally} during the full update cycle, and the reads of each block are always consistent.

In order to ensure the consistent read of $\xbm$, we leverage the dual-memory strategy for block coordinate settings proposed in \cite{Peng.etal2015} (see section 1.2.1 \emph{`Block coordinate'}).
Its key idea is that, before every write to a block $\xbm_i$, a copy of the old version of the block is kept for reading.
In this way, there always exists some state of $\xbm$ in the memory for the cores to access.

\section{Proof of Analysis}
\label{Sup:Proof}

In this section, we first present the proof of Theorem~\ref{Theo:ConvergenceBG}, then  followed by the proof of Theorem~\ref{Theo:ConvergenceSG}. For a review of monotone operators, we refer to Supplement~\ref{Sup:Review}.

Throughout the proof, we consider the probability space $(\Omega,\Fcal,P)$, where $\Omega$ denotes the sample space, $\Fcal$ the $\sigma$-algebra, and $P$ the probability measure. $\xbm^k$ is a random variable defined in $\R^n$. We use $\|\cdot\|$ to denote the $\ell_2$-norm.
We define the sequence of sub $\sigma$-algebra $\{\Xcal^{k}\}_{k\in\N}$ of $\Fcal$ as
$$\Xcal^{k}\defn\sigma(\xbm^0,...,\xbm^{k}, \Delta_0,...,\Delta_k),$$
where $\sigma$ generates the filtration (smallest $\sigma$-algebra) from $\xbm^0,...,\xbm^{k}$, and $\Delta_0,...,\Delta_k$. Note that the sequence $\{\Xcal^k\}_{k\in\N}$ is such that $\Xcal^{k} \subset \Xcal^{k+1}$ for any $k \in \N$. We use $\xbmast$ to denote some fixed point in the set $\zer(\Gsf)$.

\subsection{Proof of Theorem~\ref{Theo:ConvergenceBG}}

Our proof needs the following lemma on the RED operator.
\begin{lemma}
\label{Le:Cocoercivity}
Let Assumption~\ref{As:DataFitConvexity} and~\ref{As:NonexpansiveDen} hold for $g$ and $\Dsf_\sigma$. The composite operator $\Gsf$ is $1/(L+2\tau)$-cocoercive, that is
$$\left(\Gsf(\xbm)-\Gsf(\ybm)\right)^\Tsf\left(\xbm-\ybm\right)\geq\frac{1}{L+2\tau}\|\Gsf(\xbm)-\Gsf(\ybm)\|^2.$$
\begin{proof}
This lemma is adapted from Lemma 3 in~\cite{Sun.etal2019c}.
Consider the following decomposition
\begin{equation}
\Isf-\frac{2}{L + 2\tau}\Gsf=(\frac{2}{L + 2\tau} \cdot \frac{L}{2})\left[\Isf - \frac{2}{L} \nabla g \right]+(\frac{2}{L + 2\tau} \cdot \frac{2\tau}{2}) \left[\Isf - \frac{1}{\tau}\Hsf \right],
\end{equation}
where we recall $\Hsf=\tau(\Isf-\Dsf_\sigma)$. According to Assumption~\ref{As:DataFitConvexity}, $g$ is convex and $\nabla g$ is $L$-Lipschitz continuous. By Proposition~\ref{Prop:BlockCocoer} in Supplement~\ref{Sup:Review}, $\nabla g$ is $1/L$-cocoercive. 
Hence, by Proposition~\ref{Prop:NonexpEquiv} in Supplement~\ref{Sup:Review}, $\Isf - (2/L) \nabla g$ is nonexpansive.
Since $\Dsf_\sigma=\Isf - (1/\tau)\Hsf$, this means that $\Isf - (1/\tau)\Hsf$ is nonexpansive.
From Proposition~\ref{Prop:AveragedEquiv} in Supplement~\ref{Sup:Review}, we know that the convex combination of two nonexpansive operators is nonexpansive. Thus, $\Isf-(2/(L + 2\tau))\Gsf$ is nonexpansive, which also means that $\Gsf$ is $1/(L+2\tau)$-cocoercive according to Proposition~\ref{Prop:NonexpEquiv} in Supplement~\ref{Sup:Review}.
\end{proof}
\end{lemma}

Now we can start the main proof.
Under the fixed stepsize $\gamma>0$, we begin with the following equations regarding the fixed point $\xbmast\in\zer(\Gsf)$
\begin{align}
\label{Eq:one}
&\E\left[\|\xbm^{k+1}-\xbmast\|^2|\Xcal^{k}\right] \nonumber\\
&=\E\left[\|\xbm^{k}-\gamma\Gsf_i(\xbmtilde^{k})-\xbmast\|^2|\Xcal^{k}\right] \nonumber\\
&=\E\left[\|\xbm^{k}-\xbmast\|^2|\Xcal^{k}\right]+\gamma^2\E\left[\|\Gsf_i(\xbmtilde^{k})\|^2|\Xcal^{k}\right]+2\gamma\E\left[(\Gsf_i(\xbmtilde^{k}))^\Tsf(\xbmast-\xbm^{k})|\Xcal^{k}\right]
\end{align}
Since $\Gsf_i:\R^n\rightarrow\R^n$ is evaluated on a random block of $\xbm_i$, we have the following conditional expectations
\begin{equation}
\label{Eq:ExpOfCross1}
\E\left[(\Gsf_i(\xbmtilde^{k}))^\Tsf(\xbmast-\xbm^{k})|\Xcal^{k}\right]=\frac{1}{b}\sum_{i=1}^{b}(\Gsf_i(\xbmtilde^{k}))^\Tsf(\xbmast-\xbm^{k})=\frac{1}{b}(\Gsf(\xbmtilde^{k}))^\Tsf(\xbmast-\xbm^{k})
\end{equation}
and
\begin{equation}
\label{Eq:ExpOfNorm1}
\E\left[\|\Gsf_i(\xbmtilde^{k})\|^2|\Xcal^{k}\right]=\frac{1}{b}\sum_{i=1}^{b}\|\Gsf_i(\xbmtilde^{k})\|^2=\frac{1}{b}\|\Gsf(\xbmtilde^{k})\|^2.
\end{equation}
Thus, plugging the above results into~\eqref{Eq:one}
\begin{equation}
\E\left[\|\xbm^{k+1}-\xbmast\|^2|\Xcal^{k}\right] \leq\|\xbm^{k}-\xbmast\|^2+\frac{\gamma^2}{b}\|\Gsf(\xbmtilde^k)\|^2+\underbrace{\frac{2\gamma}{b}(\Gsf(\xbmtilde^{k}))^\Tsf(\xbmast-\xbm^{k})}_{(\dag)}.
\end{equation}
The term $(\dag)$ can be expressed as
\begin{align}
\label{Eq:two}
&\frac{2\gamma}{b}(\Gsf(\xbmtilde^{k}))^\Tsf(\xbmast-\xbm^{k}) \nonumber\\
&=\frac{2\gamma}{b}(\Gsf(\xbmtilde^{k}))^\Tsf(\xbmast-\xbmtilde^{k} + \sum_{s = k-\Delta_{k}}^{k-1}(\xbm^s - \xbm^{s+1})) \nonumber\\
&=\frac{2\gamma}{b}(\Gsf(\xbmtilde^{k})-\Gsf(\xbmast))^\Tsf(\xbmast-\xbmtilde^{k})+\frac{2\gamma}{b}(\Gsf(\xbmtilde^{k}))^\Tsf(\sum_{s = k-\Delta_{k}}^{k-1}(\xbm^s - \xbm^{s+1}))\nonumber\\
&=\frac{2\gamma}{b}(\Gsf(\xbmtilde^{k})-\Gsf(\xbmast))^\Tsf(\xbmast-\xbmtilde^{k})+\frac{2\gamma^2}{b}\sum_{s = k-\Delta_{k}}^{k-1}\Gsf(\xbmtilde^{k})^\Tsf\Gsf_{i_s}(\xbmtilde^s),
\end{align}
where in the second line we used the definition of the stale iterate $\xbm^{s+1}=\xbm^s-\gamma\Gsf_{i_s}(\xbmtilde^k)$, and in the third line the fact that $\Gsf(\xbmast)=\mathbf{0}$. By using Lemma~\ref{Le:Cocoercivity}, we obtain the upper bound for the first term in equation~\eqref{Eq:two}
\begin{equation}
\label{Eq:twoA}
\frac{2\gamma}{b}(\Gsf(\xbmtilde^{k})-\Gsf(\xbmast))^\Tsf(\xbmast-\xbmtilde^{k}) \leq -\frac{2\gamma\|\Gsf(\xbmtilde^{k})\|^2}{b(L+2\tau)}.
\end{equation}
For the second term in~\eqref{Eq:two}, we have
\begin{align}
\label{Eq:twoB}
\frac{2\gamma^2}{b}\sum_{s = k-\Delta_{k}}^{k-1}\Gsf(\xbmtilde^{k})^\Tsf\Gsf_{i_s}(\xbmtilde^s)&\leq\frac{\lambda\gamma^2\|\Gsf(\xbmtilde^{k})\|^2}{b}+\sum_{s = k-\Delta_{k}}^{k-1}\frac{\gamma^2\|\Gsf_{i_s}(\xbmtilde^s)\|^2}{b}, \nonumber\\
&\leq\frac{\lambda\gamma^2\|\Gsf(\xbmtilde^{k})\|^2}{b}+\sum_{s = k-\lambda}^{k-1}\frac{\gamma^2\|\Gsf(\xbmtilde^s)\|^2}{b},
\end{align}
where in the first inequality we used the \emph{Young's inequality}
\begin{equation}
\label{Eq:Young}
\xbm_1^\Tsf\xbm_2\leq\frac{1}{2}\left[\|\xbm_1\|^2+\|\xbm_2\|^2\right],
\end{equation}
and in the second inequality we use
$$\sum_{s = k-\Delta k}^{k-1}\gamma^2\|\Gsf_{i_s}(\xbmtilde^{s})\|^2=\sum_{s = k-\Delta k}^{k-1}\|\xbm^s - \xbm^{s+1}\|_2^2\leq\sum_{s=k-\lambda}^{k-1}\|\xbm^s - \xbm^{s+1}\|_2^2=\sum_{s = k-\lambda}^{k-1}\gamma^2\|\Gsf(\xbmtilde^{s})\|^2.$$
Applying~\eqref{Eq:twoA} and~\eqref{Eq:twoB} in~\eqref{Eq:two} yields the overall upper bound for the term $(\dag)$
\begin{equation}
\label{Eq:three}
\frac{2\gamma}{b}(\Gsf(\xbmtilde^{k}))^\Tsf(\xbmast-\xbm^{k})\leq\frac{(L+2\tau)\lambda\gamma^2-2\gamma}{(L+2\tau)b}\|\Gsf(\xbmtilde^{k})\|^2+\sum_{s = k-\lambda}^{k-1}\frac{\gamma^2\|\Gsf(\xbmtilde^s)\|^2}{b}.
\end{equation}
Next, by plugging~\eqref{Eq:three} into~\eqref{Eq:one} and re-arranging the terms, we obtain the following inequality
\begin{align}
\label{Eq:FundamentalIneq1}
&\E\left[\|\xbm^{k+1}-\xbmast\|^2|\Xcal^{k}\right] \nonumber\\
&\leq\|\xbm^{k}-\xbmast\|^2+\sum_{s = k-\lambda}^{k-1}\frac{\gamma^2\|\Gsf(\xbmtilde^s)\|^2}{b}+\frac{(L+2\tau)(1+\lambda)\gamma^2-2\gamma}{(L+2\tau)b}\|\Gsf(\xbmtilde^{k})\|^2.
\end{align}
Taking the total expectation of equation~\eqref{Eq:FundamentalIneq1} and re-arranging the terms yields that
\begin{align}
\label{Eq:four}
&\frac{2\gamma-(L+2\tau)(1+\lambda)\gamma^2}{(L+2\tau)b}\E\left[\|\Gsf(\xbmtilde^{k})\|^2\right] \nonumber\\
&\leq\E\left[\|\xbm^{k}-\xbmast\|^2\right]-\E\left[\|\xbm^{k+1}-\xbmast\|^2\right]+\gamma^2\sum_{s = k-\lambda}^{k-1}\frac{\E\left[\|\Gsf(\xbmtilde^{s})\|^2\right]}{b}
\end{align}
We then telescope-sum equation~\eqref{Eq:four} over $t>0$ iterations to have
\begin{align}
\label{Eq:five}
&\sum_{k=0}^{t-1}\frac{2\gamma-(L+2\tau)(1+\lambda)\gamma^2}{(L+2\tau)b}\E\left[\|\Gsf(\xbmtilde^{k})\|^2\right] \nonumber\\
&\leq\E\left[\|\xbm^{0}-\xbmast\|^2\right]-\E\left[\|\xbm^{t}-\xbmast\|^2\right]+\gamma^2\sum_{k = 0}^{t-1}\sum_{s = k-\lambda}^{k-1}\frac{\E\left[\|\Gsf(\xbmtilde^{s})\|^2\right]}{b}
\end{align}
where the index $s$ always start at $0$. Under the assumption of consistent read, it is true that
\begin{equation}
\label{Eq:DoubleSum}
\sum_{k = 0}^{t-1}\sum_{s = k-\lambda}^{k-1}\frac{\E\left[\|\Gsf(\xbmtilde^{s})\|^2\right]}{b}\leq\lambda\sum_{k = 0}^{t-1}\frac{\E\left[\|\Gsf(\xbmtilde^{k})\|^2\right]}{b}.
\end{equation}
In the case of inconsistent read, the above inequality does not always hold. 
We refer to \cite{Peng.etal2015} for a comprehensive analysis for asynchronous block-coordinate methods with inconsistent reads.
Now, we rewrite equation~\eqref{Eq:five} as
\begin{align}
\label{Eq:six}
&\sum_{k=0}^{t-1}\frac{2\gamma-(L+2\tau)(1+2\lambda)\gamma^2}{(L+2\tau)b}\E\left[\|\Gsf(\xbmtilde^{k})\|^2\right]\leq\E\left[\|\xbm^{0}-\xbmast\|^2\right]-\E\left[\|\xbm^{t}-\xbmast\|^2\right].
\end{align}
In order to ensure the convergence, we need the coefficient of $\E\left[\|\Gsf(\xbmtilde^{k})\|^2\right]$ to be positive. From basic algebra, one feasible range for the stepsize $\gamma$ is
$$0<\gamma\leq\frac{1}{(L+2\tau)(1+2\lambda)},$$
which directly implies that
\begin{equation}
0<\frac{\gamma}{(L+2\tau)b}\leq\frac{2\gamma-(L+2\tau)(1+2\lambda)\gamma^2}{(L+2\tau)b}.\nonumber
\end{equation}
By simplifying~\eqref{Eq:six} with the above result and dropping the negative term, we can derive the following bound for the $\E\left[\|\Gsf(\xbmtilde^k)\|^2\right]$ averaged over $t$ iterations
\begin{equation}
\label{Eq:Delay1}
\frac{1}{t}\sum_{k=0}^{t-1}\E\left[\|\Gsf(\xbmtilde^{k})\|^2\right]\leq\frac{(L+2\tau)b}{\gamma t}\E\left[\|\xbm^{0}-\xbmast\|^2\right]\leq\frac{(L+2\tau)b}{\gamma t}R_0^2.
\end{equation}
The above inequality establishes that the change of the stale iterate $\xbmtilde^{k}$ converges to zero as $t$ increases.
Next, we will use the bound to establish the similar result for the actual iterate $\xbm^k$.
We know that $\|\Gsf(\xbm^k)\|^2$ can be bounded by
\begin{align}
\label{Eq:oone}
\|\Gsf(\xbm^k)\|^2 &\leq(\|\Gsf(\xbm^k)-\Gsf(\xbmtilde^k)\|+\|\Gsf(\xbmtilde^k)\|)^2 \nonumber\\
&= \|\Gsf(\xbm^k)-\Gsf(\xbmtilde^k)\|^2 + \|\Gsf(\xbmtilde^k)\|^2+2\|\Gsf(\xbm^k)-\Gsf(\xbmtilde^k)\|\|\Gsf(\xbmtilde^k)\|\nonumber\\
&\leq 2\|\Gsf(\xbm^k)-\Gsf(\xbmtilde^k)\|^2 + 2\|\Gsf(\xbmtilde^k)\|^2 \nonumber\\
&\leq2(L+2\tau)^2\|\xbm^k-\xbmtilde^k\|^2 + 2\|\Gsf(\xbmtilde^k)\|^2
\end{align}
where in the second inequality we used the Young's inequality~\eqref{Eq:Young}, and in the third inequality we used the following result implied by Lemma~\ref{Le:Cocoercivity}
$$(L+2\tau)\|\xbm-\ybm\|\geq\|\Gsf(\xbm)-\Gsf(\ybm)\|.$$
By expressing the stale iterate $\xbmtilde^k$, we can write equation~\eqref{Eq:oone} as
\begin{align}
\label{Eq:oonee}
\|\Gsf(\xbm^k)\|^2&\leq2(L+2\tau)^2\|\sum^{k-1}_{s=k-\lambda}\gamma\Gsf_{i_s}(\xbmtilde^s)\|^2+2\|\Gsf(\xbmtilde^k)\|^2. \nonumber\\
&\leq2\lambda(L+2\tau)^2\sum^{k-1}_{s=k-\lambda}\gamma^2\|\Gsf_{i_s}(\xbmtilde^s)\|^2+2\|\Gsf(\xbmtilde^k)\|^2.
\end{align}
where we use the fact
$$\|\sum_{i=1}^{n}\xbm_i\|^2=\sum_{i=1}^{n}\|\xbm_i\|^2+\sum_{a\neq b}\xbm_a^\Tsf\xbm_b\leq\sum_{i=1}^{n}\|\xbm_i\|^2+\frac{1}{2}\sum_{a\neq b}\left[\|\xbm_a\|^2+\|\xbm_b\|^2\right]=n\sum_{i=1}^{n}\|\xbm_i\|^2$$
Taking the expectation of equation~\eqref{Eq:oonee} leads to
\begin{align}
\label{Eq:ttwo}
&\E\left[\|\Gsf(\xbm^k)\|^2\right] \nonumber\\
&\leq2\lambda(L+2\tau)^2\sum^{k-1}_{s=k-\lambda}\gamma^2\E\left[\|\Gsf_{i_s}(\xbmtilde^s)\|^2\right]+2\E\left[\|\Gsf(\xbmtilde^k)\|^2\right] \nonumber\\
&\leq2\lambda(L+2\tau)^2\sum^{k-1}_{s=k-\lambda}\frac{\gamma^2\E\left[\|\Gsf(\xbmtilde^{s})\|^2\right]}{b}+2\E\left[\|\Gsf(\xbmtilde^k)\|^2\right],
\end{align}
By averaging~\eqref{Eq:ttwo} over $t>0$ iterations, we obtain that
\begin{align}
\label{Eq:tthree}
&\frac{1}{t}\sum^{t-1}_{k=0}\E\left[\|\Gsf(\xbm^k)\|^2\right] \nonumber\\
&\leq\frac{2\lambda(L+2\tau)^2}{t}\sum^{t-1}_{k=0}\sum^{k-1}_{s=k-\lambda}\frac{\gamma^2\E\left[\|\Gsf(\xbmtilde^{s})\|^2\right]}{b}+ \frac{2}{t}\sum^{t-1}_{k=0}\E\left[\|\Gsf(\xbmtilde^k)\|^2\right] \nonumber\\
&\leq\frac{2\lambda^2(L+2\tau)^2}{t}\sum^{t-1}_{k=0}\frac{\gamma^2\E\left[\|\Gsf(\xbmtilde^{k})\|^2\right]}{b}+\frac{2}{t}\sum^{t-1}_{k=0}\E\left[\|\Gsf(\xbmtilde^k)\|^2\right]
\end{align}
where we again used result in~\eqref{Eq:DoubleSum} in the last inequality.
Re-arranging the terms in~\eqref{Eq:tthree} yields
\begin{align}
\label{Eq:ffour}
\frac{1}{t}\sum^{t-1}_{k=0}\E\left[\|\Gsf(\xbm^k)\|^2\right]\leq\left[\frac{2\lambda^2(L+2\tau)^2}{b}\gamma^2+2\right]\frac{1}{t}\sum^{t-1}_{k=0}\E\left[\|\Gsf(\xbmtilde^k)\|^2\right]
\end{align}
We plug the result in~\eqref{Eq:Delay1} into \eqref{Eq:ffour} and obtain
\begin{align}
\label{Eq:ffive}
\frac{1}{t}\sum^{t-1}_{k=0}\E\left[\|\Gsf(\xbm^k)\|^2\right]
\leq\left[\frac{2\lambda^2(L+2\tau)^2}{b}\gamma^2+2\right]\frac{(L+2\tau)b}{\gamma t}R_0^2,
\end{align}
Since it is always true that 
$$\gamma\leq\frac{1}{(L+2\tau)(1+2\lambda)}\leq\frac{1}{(L+2\tau)(1+\lambda)}.$$
we can simplify the bound by using the above inequality related to the stepsize $\gamma$
\begin{align}
\label{Eq:SSix}
\frac{1}{t}\sum^{t-1}_{k=0}\E\left[\|\Gsf(\xbm^k)\|^2\right]
\leq\left[\frac{2\lambda^2}{(1+\lambda)^2b}+2\right]\frac{(L+2\tau)b}{\gamma t}R_0^2.
\end{align}
Let $D=2\lambda^2/(1+\lambda)^2$, and we derive the desired result.
\begin{equation}
\label{Eq:SSeven}
\min_{0\leq k\leq t-1}\E\left[\|\Gsf(\xbm^k)\|^2\right]\leq\frac{1}{t}\sum^{t-1}_{k=0}\E\left[\|\Gsf(\xbm^k)\|^2\right]\leq\left[\frac{D}{b}+2\right]\frac{(L+2\tau)b}{\gamma t}R_0^2.
\end{equation}

\subsection{Proof of Theorem~\ref{Theo:ConvergenceSG}}

We prove Theorem~\ref{Theo:ConvergenceSG} by following the procedure in the proof of Theorem~\ref{Theo:ConvergenceBG} with the adaptation to the block stochastic operator $\Gsfhat_i$.
In the key steps, we will highlight the difference between the two proofs.
In addition to Lemma~\ref{Le:Cocoercivity}, our second proof requires the following lemma related to the statistical properties of $\Gsfhat$.
\begin{lemma}
\label{Le:Variance}
Let Assumption~\ref{As:DataFitConvexity} and~\ref{As:NonexpansiveDen} hold for $g$ and $\Dsf_\sigma$. Then, we can establish the following statements for operator $\Gsfhat$
$$\E\left[\Gsfhat(\xbm)\right]=\Gsf(\xbm),\quad\E\left[\|\Gsfhat(\xbm)-\Gsf(\xbm)\|^2\right]\leq\frac{\nu^2}{w},$$
which further implies that
$$\E\left[\|\Gsfhat(\xbm)\|^2\right]\leq\frac{\nu^2}{w}+\|\Gsf(\xbm)\|^2.$$
\begin{proof}
Since the the stochasticity happens only in the evaluation of the gradient, it is straightforward to see that
\begin{equation}
\E\left[\Gsfhat(\xbm)\right] = \E[\widehat{\nabla} g(\xbm)] + \Dsf_\sigma(\xbm) = \Gsf(\xbm), \nonumber
\end{equation}
Similarly, we have that
$$\E\left[\|\Gsfhat(\xbm)-\Gsf(\xbm)\|_2^2\right]=\E\left[\|\widehat{\nabla} g(\xbm)-\nabla g(\xbm)\|_2^2\right]\leq\frac{\nu^2}{w}$$
Given that $\Tr(\E\left[X^\Tsf X\right])=\Tr(\Cov\left[X\right])+\Tr(\E\left[X\right]^2)$, we obtain that
\begin{equation}
\E\left[\|\Gsfhat(\xbm)\|^2\right]=\E\left[\|\Gsfhat(\xbm)-\Gsf(\xbm)\|^2\right]+\E\left[\Gsfhat(\xbm)\right]^2 \leq\frac{\nu^2}{w}+\|\Gsf(\xbm)\|^2, \nonumber
\end{equation}
where we let $\E\left[\Gsfhat(\xbm)\right]^2\defn\E\left[\Gsfhat(\xbm)\right]^\Tsf\E\left[\Gsfhat(\xbm)\right]$. Note that $\Tr(\cdot)$ and $\Cov(\cdot)$ denote the computation of the trace and covariance of a matrix and a vector, respectively.
\end{proof}
\end{lemma}
Now we start the proof. Similar as~\eqref{Eq:one}, we write that
\begin{align}
\label{Eq:One}
&\E\left[\|\xbm^{k+1}-\xbmast\|^2|\Xcal^{k}\right] \nonumber\\
&=\E\left[\|\xbm^{k}-\gamma\Gsfhat_i(\xbmtilde^{k})-\xbmast\|^2|\Xcal^{k}\right] \nonumber\\
&=\E\left[\|\xbm^{k}-\xbmast\|^2|\Xcal^{k}\right]+\gamma^2\E\left[\|\Gsfhat_i(\xbmtilde^{k})\|^2|\Xcal^{k}\right]+2\gamma\E\left[(\Gsfhat_i(\xbmtilde^{k}))^\Tsf(\xbmast-\xbm^{k})|\Xcal^{k}\right]
\end{align}
Here, the conditional expectation is taken for $\Gsfhat_i(\xbm)=\Usf_i\Usf_i^\Tsf\Gsfhat(\xbm)$. By using Lemma~\ref{Le:Variance}, we can compute conditional expectations as
\begin{equation}
\label{Eq:ExpOfCross2}
\E\left[(\Gsfhat_i(\xbmtilde^{k}))^\Tsf(\xbmast-\xbm^{k})|\Xcal^{k}\right]=\frac{1}{b}\E\left[(\Gsfhat(\xbmtilde^{k}))^\Tsf(\xbmast-\xbm^{k})|\Xcal^{k}\right]=\frac{1}{b}(\Gsf(\xbmtilde^{k}))^\Tsf(\xbmast-\xbm^{k})
\end{equation}
and
\begin{equation}
\label{Eq:ExpOfNorm2}
\E\left[\|\Gsfhat_i(\xbmtilde^{k})\|^2|\Xcal^{k}\right]=\frac{1}{b}\E\left[\|\Gsfhat(\xbmtilde^{k})\|^2|\Xcal^{k}\right] \leq\frac{\nu^2}{wb}+\frac{\|\Gsf(\xbmtilde^{k})\|^2}{b}.
\end{equation}
where we first compute the expectation corresponding to the randomized block and then the expectation for the stochastic measurements.
We note that the expectation of the cross term~\eqref{Eq:ExpOfCross2} remains the same as the result in~\eqref{Eq:ExpOfCross1}, while the expectation in~\eqref{Eq:ExpOfNorm2} has one extra term related to the norm variance of the stochastic operator compared with~\eqref{Eq:ExpOfNorm1}.
As we shall see in the future steps, the difference in the expectation of the operator's squared norm leads to the most modifications.
Using the above results in equation~\eqref{Eq:One} yields that
\begin{align}
&\E\left[\|\xbm^{k+1}-\xbmast\|^2|\Xcal^{k}\right] \nonumber\\
&\leq\|\xbm^{k}-\xbmast\|^2+\frac{\gamma^2}{b}\|\Gsf(\xbmtilde^k)\|^2+\frac{\gamma^2\nu^2}{wb} +\underbrace{\frac{2\gamma}{b}(\Gsf(\xbmtilde^{k}))^\Tsf(\xbmast-\xbm^{k})}_{(\dag)}.
\end{align}
By following~\eqref{Eq:two}, we can express the term $(\dag)$ as
\begin{align}
\label{Eq:Two}
&\frac{2\gamma}{b}(\Gsf(\xbmtilde^{k}))^\Tsf(\xbmast-\xbm^{k}) \nonumber\\
&=\frac{2\gamma}{b}(\Gsf(\xbmtilde^{k})-\Gsf(\xbmast))^\Tsf(\xbmast-\xbmtilde^{k})+\frac{2\gamma^2}{b}\sum_{s = k-\Delta_{k}}^{k-1}\Gsf(\xbmtilde^{k})^\Tsf\Gsfhat_{i_s}(\xbmtilde^s),
\end{align}
The upper bound of the first term is the same as shown in~\eqref{Eq:twoA}, which is
\begin{equation}
\label{Eq:TwoA}
\frac{2\gamma}{b}(\Gsf(\xbmtilde^{k})-\Gsf(\xbmast))^\Tsf(\xbmast-\xbmtilde^{k}) \leq -\frac{2\gamma\|\Gsf(\xbmtilde^{k})\|^2}{b(L+2\tau)}.
\end{equation}
Similarly, our second term is bounded by
\begin{align}
\label{Eq:TwoB}
\frac{2\gamma^2}{b}\sum_{s = k-\Delta_{k}}^{k-1}\Gsf(\xbmtilde^{k})^\Tsf\Gsfhat_{i_s}(\xbmtilde^s)
&\leq\frac{\lambda\gamma^2\|\Gsf(\xbmtilde^{k})\|^2}{b}+\sum_{s = k-\lambda}^{k-1}\frac{\gamma^2\|\Gsfhat(\xbmtilde^s)\|^2}{b},
\end{align}
where we used the Young's inequality~\eqref{Eq:Young} together with the fact that
\begin{equation}
\sum_{s = k-\Delta k}^{k-1}\|\Gsfhat_{i_s}(\xbmtilde^{k})\|^2\leq\sum_{s = k-\lambda}^{k-1}\|\Gsfhat_{i_s}(\xbmtilde^{k})\|^2\leq\sum_{s = k-\lambda}^{k-1}\|\Gsfhat(\xbmtilde^{k})\|^2. \nonumber
\end{equation}
Equation~\eqref{Eq:TwoA} and~\eqref{Eq:TwoB} together establish the overall upper bound for the term $(\dag)$
\begin{equation}
\label{Eq:Three}
\frac{2\gamma}{b}(\Gsf(\xbmtilde^{k}))^\Tsf(\xbmast-\xbm^{k})\leq\frac{(L+2\tau)\lambda\gamma^2-2\gamma}{(L+2\tau)b}\|\Gsf(\xbmtilde^{k})\|^2+\sum_{s = k-\lambda}^{k-1}\frac{\gamma^2\|\Gsfhat(\xbmtilde^s)\|^2}{b}.
\end{equation}
By plugging~\eqref{Eq:Three} into~\eqref{Eq:One} and re-arranging the terms, we obtain that
\begin{align}
\label{Eq:FundamentalIneq}
&\E\left[\|\xbm^{k+1}-\xbmast\|^2|\Xcal^{k}\right] \nonumber\\
&\leq\|\xbm^{k}-\xbmast\|^2+\frac{\gamma^2\nu^2}{wb}+\sum_{s = k-\lambda}^{k-1}\frac{\gamma^2\|\Gsfhat(\xbmtilde^s)\|^2}{b}+\frac{(L+2\tau)(1+\lambda)\gamma^2-2\gamma}{(L+2\tau)b}\|\Gsf(\xbmtilde^{k})\|^2.
\end{align}
Taking the total expectation of equation~\eqref{Eq:FundamentalIneq} and re-arranging the terms yields that
\begin{align}
\label{Eq:Four}
&\frac{2\gamma-(L+2\tau)(1+\lambda)\gamma^2}{(L+2\tau)b}\E\left[\|\Gsf(\xbmtilde^{k})\|^2\right] \nonumber\\
&\leq\E\left[\|\xbm^{k}-\xbmast\|^2\right]-\E\left[\|\xbm^{k+1}-\xbmast\|^2\right]+\frac{\gamma^2\nu^2}{wb}+\gamma^2\sum_{s = k-\lambda}^{k-1}\left[\frac{\nu^2}{wb}+\frac{\E\left[\|\Gsf(\xbmtilde^{s})\|^2\right]}{b}\right]
\end{align}
where we use the following inequality derived by using the law of total expectation and Lemma~\ref{Le:Variance}
\begin{align}
\label{Eq:TotalExp}
\E\left[\|\Gsfhat(\xbmtilde^{s})\|^2\right]&=\E\left[\E\left[\|\Gsfhat(\xbmtilde^{s})\|^2|\Xcal^{s}\right]\right]\leq\frac{\nu^2}{w}+\E\left[\|\Gsf(\xbmtilde^{s})\|^2\right].
\end{align}
We telescope-sum equation~\eqref{Eq:Four} over $t>0$ iterations to obtain
\begin{align}
\label{Eq:Five}
&\sum_{k=0}^{t-1}\frac{2\gamma-(L+2\tau)(1+\lambda)\gamma^2}{(L+2\tau)b}\E\left[\|\Gsf(\xbmtilde^{k})\|^2\right] \nonumber\\
&\leq\E\left[\|\xbm^{0}-\xbmast\|^2\right]-\E\left[\|\xbm^{t}-\xbmast\|^2\right]+\sum_{k = 0}^{t-1}\frac{\gamma^2\nu^2}{wb}+\gamma^2\sum_{k = 0}^{t-1}\sum_{s = k-\lambda}^{k-1}\left[\frac{\nu^2}{wb}+\frac{\E\left[\|\Gsf(\xbmtilde^{s})\|^2\right]}{b}\right]
\end{align}
By applying the same relaxation trick in~\eqref{Eq:DoubleSum} to~\eqref{Eq:Five}
\begin{align}
\label{Eq:DoubleSum2}
&\sum_{k = 0}^{t-1}\sum_{s = k-\lambda}^{k-1}\left[\frac{\nu^2}{wb}+\frac{\E\left[\|\Gsf(\xbmtilde^{s})\|^2\right]}{b}\right]\leq\lambda\sum_{k = 0}^{t-1}\left[\frac{\nu^2}{wb}+\frac{\E\left[\|\Gsf(\xbmtilde^{k})\|^2\right]}{b}\right],
\end{align}
we then have that
\begin{align}
\label{Eq:Six}
&\sum_{k=0}^{t-1}\frac{2\gamma-(L+2\tau)(1+2\lambda)\gamma^2}{(L+2\tau)b}\E\left[\|\Gsf(\xbmtilde^{k})\|^2\right]\leq\E\left[\|\xbm^{0}-\xbmast\|^2\right]+\frac{(1+\lambda)\gamma^2\nu^2}{wb}\cdot t,
\end{align}
where we dropped the negative term. Recall that if $\gamma$ is in the range $\gamma\in(0, 1/((L+2\tau)(1+2\lambda))]$, we have the inequality
\begin{equation}
\frac{\gamma}{(L+2\tau)b}\leq\frac{2\gamma-(L+2\tau)(1+2\lambda)\gamma^2}{(L+2\tau)b}.\nonumber
\end{equation}
By relaxing the coefficient in the lefthand side, dividing the inequality by $t$, and re-arranging the terms, we obtain the convergence in terms of the stale iterate $\xbmtilde^k$
\begin{align}
\label{Eq:Delay2}
\frac{1}{t}\sum_{k=0}^{t-1}\E\left[\|\Gsf(\xbmtilde^{k})\|^2\right]&\leq\frac{(L+2\tau)b}{\gamma t}\left[\E\left[\|\xbm^{0}-\xbmast\|^2\right]+\frac{(1+\lambda)\gamma^2\nu^2}{wb}\cdot t\right] \nonumber\\
&\leq\frac{(L+2\tau)b}{\gamma t}R_0^2+\frac{\gamma}{w}C
\end{align}
where we used Assumption~\ref{As:NonemptySet} and let $C = (L+2\tau)(1+\lambda)\nu^2$.
Compared with the result in equation~\eqref{Eq:Delay1},equation~\eqref{Eq:Delay2} has the extra term related to the variance of $\Gsfhat_i(\xbm)$. Next, we establish the convergence in terms of actual iterate $\xbm^k$. Following the steps from~\eqref{Eq:oone} to~\eqref{Eq:ttwo}, we directly obtain the inequality related to $\Gsfhat_i(\xbmtilde)$
\begin{equation}
\label{Eq:TTwo}
\E\left[\|\Gsf(\xbm^k)\|^2\right]\leq2\lambda(L+2\tau)^2\sum^{k-1}_{s=k-\lambda}\gamma^2\E\left[\|\Gsfhat_{i_s}(\xbmtilde^s)\|^2\right]+2\E\left[\|\Gsf(\xbmtilde^k)\|^2\right]
\end{equation}
By using the the result in~\eqref{Eq:TotalExp}, we derive from~\eqref{Eq:TTwo} that
\begin{equation}
\E\left[\|\Gsf(\xbm^k)\|^2\right]\leq2\lambda(L+2\tau)^2\sum^{k-1}_{s=k-\lambda}\gamma^2\left[\frac{\nu^2}{wb}+\frac{\E\left[\|\Gsf(\xbmtilde^{s})\|^2\right]}{b}\right]+2\E\left[\|\Gsf(\xbmtilde^k)\|^2\right].
\end{equation}
By averaging~\eqref{Eq:TTwo} over $t>0$ iterations, we obtain that
\begin{align}
\label{Eq:TThree}
&\frac{1}{t}\sum^{t-1}_{k=0}\E\left[\|\Gsf(\xbm^k)\|^2\right] \nonumber\\
&\leq\frac{2\lambda(L+2\tau)^2}{t}\sum^{t-1}_{k=0}\sum^{k-1}_{s=k-\lambda}\gamma^2\left[\frac{\nu^2}{wb}+\frac{\E\left[\|\Gsf(\xbmtilde^{s})\|^2\right]}{b}\right]+ \frac{2}{t}\sum^{t-1}_{k=0}\E\left[\|\Gsf(\xbmtilde^k)\|^2\right] \nonumber\\
&\leq\frac{2\lambda^2(L+2\tau)^2}{t}\sum^{t-1}_{k=0}\gamma^2\left[\frac{\nu^2}{wb}+\frac{\E\left[\|\Gsf(\xbmtilde^{k})\|^2\right]}{b}\right]  + \frac{2}{t}\sum^{t-1}_{k=0}\E\left[\|\Gsf(\xbmtilde^k)\|^2\right]
\end{align}
where we again used the relaxation~\eqref{Eq:DoubleSum2} in the last inequality.
Re-arranging the terms in~\eqref{Eq:TThree} yields
\begin{align}
\label{Eq:FFour}
&\frac{1}{t}\sum^{t-1}_{k=0}\E\left[\|\Gsf(\xbm^k)\|^2\right] \nonumber\\
&\leq\frac{2\lambda^2(L+2\tau)^2\cdot\nu^2}{wb}\gamma^2+\left[\frac{2\lambda^2(L+2\tau)^2}{b}\gamma^2+2\right]\frac{1}{t}\sum^{t-1}_{k=0}\E\left[\|\Gsf(\xbmtilde^k)\|^2\right]
\end{align}
We plug the result in~\eqref{Eq:Delay2} into \eqref{Eq:FFour} and obtain
\begin{align}
\label{Eq:FFive}
&\frac{1}{t}\sum^{t-1}_{k=0}\E\left[\|\Gsf(\xbm^k)\|^2\right] \nonumber\\
&\leq\frac{2\lambda^2(L+2\tau)^2\cdot\nu^2}{wb}\gamma^2+\left[\frac{2\lambda^2(L+2\tau)^2}{b}\gamma^2+2\right]\left[\frac{(L+2\tau)b}{\gamma t}R_0^2+\frac{\gamma}{w}C\right]
\end{align}
Similarly, we can use the fact
$$\gamma\leq\frac{1}{(L+2\tau)(1+\lambda)}.$$
to simplify the bound in~\eqref{Eq:FFive}
\begin{align}
\label{Eq:SSix}
&\frac{1}{t}\sum^{t-1}_{k=0}\E\left[\|\Gsf(\xbm^k)\|^2\right]\nonumber\\
&\leq\frac{2\lambda^2(L+2\tau)^2\cdot\nu^2}{wb}\cdot\frac{1}{(L+2\tau)(1+\lambda)} \cdot\gamma+\left[\frac{2\lambda^2}{(1+\lambda)^2b}+2\right]\left[\frac{(L+2\tau)b}{\gamma t}R_0^2+\frac{\gamma}{w}C\right]\nonumber\\
&=\frac{2\lambda^2}{(1+\lambda)^2b}\cdot\frac{(L+2\tau)(1+\lambda)\nu^2}{w}\cdot\gamma+\left[\frac{2\lambda^2}{(1+\lambda)^2b}+2\right]\left[\frac{(L+2\tau)b}{\gamma t}R_0^2+\frac{\gamma}{w}C\right]\nonumber\\
&=\frac{2\lambda^2}{(1+\lambda)^2b}\cdot\frac{C}{w}\gamma+\left[\frac{2\lambda^2}{(1+\lambda)^2b}+2\right]\left[\frac{(L+2\tau)b}{\gamma t}R_0^2+\frac{\gamma}{w}C\right]
\end{align}
where we recall $C = (L+2\tau)(1+\lambda)\nu^2$. Let $D=2\lambda^2/(1+\lambda)^2$ and we can derive the result of Theorem~\ref{Theo:ConvergenceSG}
\begin{equation}
\label{Eq:SSeven}
\min_{0\leq k\leq t-1}\E\left[\|\Gsf(\xbm^k)\|^2\right]\leq\frac{1}{t}\sum^{t-1}_{k=0}\E\left[\|\Gsf(\xbm^k)\|^2\right]\leq\left[\frac{D}{b}+2\right]\frac{(L+2\tau)b}{\gamma t}R_0^2+\left[\frac{2D}{b}+2\right]\frac{\gamma}{w}C,
\end{equation}
which immediately implies the result in remark 1 by setting $\gamma=1/\sqrt{wt}$
\begin{equation}
\min_{0\leq k\leq t-1}\E\left[\|\Gsf(\xbm^k)\|^2\right]\leq\frac{1}{t}\sum^{t-1}_{k=0}\E\left[\|\Gsf(\xbm^k)\|^2\right]\leq\left[\frac{D}{b}+2\right]\frac{(L+2\tau)b}{\sqrt{wt}}R_0^2+\left[\frac{2D}{b}+2\right]\frac{C}{\sqrt{wt}}.
\end{equation}
From basic algebra, we can derive the condition for $\lambda$
$$\frac{1}{\sqrt{wt}}\leq\frac{1}{(L+2\tau)(1+2\lambda)}\quad\Rightarrow\quad\lambda\leq\frac{1}{2}\left[\frac{\sqrt{wt}}{L+2\tau}-1\right].$$

\section{Background on Monotone Operators}
\label{Sup:Review}

The results in our review can be found in different forms in standard textbooks~\cite{Rockafellar.Wets1998, Boyd.Vandenberghe2004, Nesterov2004, Bauschke.Combettes2017}, and we include these results for completeness.

\medskip
\begin{definition}
An operator $\Tsf$ is Lipschitz continuous with constant $L > 0$ if
\begin{equation*}
\|\Tsf\xbm - \Tsf\ybm\| \leq L\|\xbm-\ybm\|,\quad \xbm, \ybm \in \R^n.
\end{equation*}
When $L = 1$, we say that $\Tsf$ is nonexpansive. When $L < 1$, we say that $\Tsf$ is a contraction.
\end{definition}

\begin{figure}[t]
\centering\includegraphics[width=0.7\linewidth]{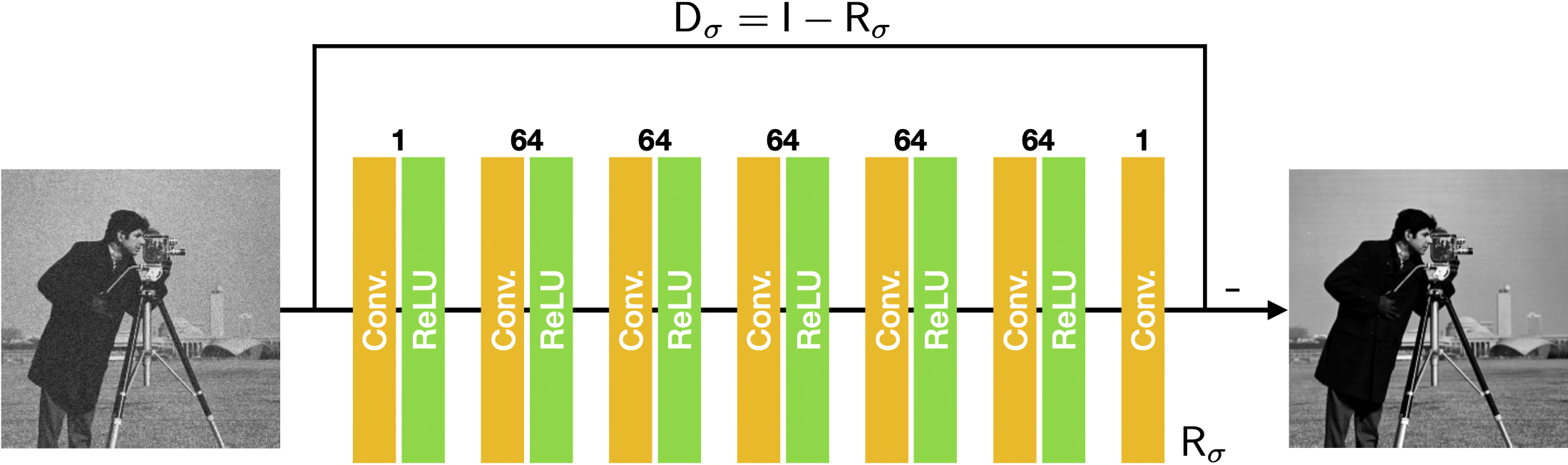}
\caption{Illustration of the architecture of DnCNN used in all experiments. The neural net is trained to remove the AWGN from its noisy input image. We also constrains the Lipschitz constant of $\Rsf_\sigma$ to be smaller than $2$ by using the spectral normalization technique in~\cite{Sedghi2018}. This provides a necessary condition for the satisfaction of Assumption~\ref{As:NonexpansiveDen}.}
\label{Fig:Architecture}
\end{figure}

\medskip
\begin{definition}
$\Tsf$ is monotone if
\begin{equation*}
(\Tsf(\xbm)-\Tsf(\ybm))^\Tsf(\xbm-\ybm) \geq 0,\quad \xbm, \ybm \in \R^n.
\end{equation*}
We say that it is strongly monotone or coercive with parameter $\mu > 0$ if
\begin{equation*}
(\Tsf(\xbm)-\Tsf(\ybm))^\Tsf(\xbm-\ybm) \geq \mu\|\xbm-\ybm\|^2,\quad \xbm, \ybm \in \R^n.
\end{equation*}
\end{definition}

\begin{definition}
$\Tsf$ is cocoercive with constant $\beta > 0$ if
\begin{equation*}
(\Tsf(\xbm)-\Tsf(\ybm))^\Tsf(\xbm-\ybm) \geq \beta\|\Tsf\xbm-\Tsf\ybm\|^2, \quad \xbm, \ybm \in \R^n.
\end{equation*}
When $\beta = 1$, we say that $\Tsf$ is firmly nonexpansive.
\end{definition}

\medskip\noindent
The following results are derived from the definition above.

\begin{proposition}
\label{Prop:BlockCocoer}
For a convex and continuously differentiable function $f$, we have
$$\nabla f \text{ is $L$-Lipschitz continuous} \quad\Leftrightarrow\quad \nabla f \text{ is  $(1/L)$-cocoercive}.$$
\end{proposition}

\begin{proof}
The proof is a minor variation of the one presented as Theorem~2.1.5 in Section~2.1 of~\cite{Nesterov2004}.
\end{proof}

%

\begin{proposition}
\label{Prop:NonexpEquiv}
Consider $\Tsf: \R^n \rightarrow \R^n$ and $\beta > 0$. Then, the following are equivalent
\begin{equation}
\Tsf \text{ is } \beta\text{-cocoercive} \quad\Leftrightarrow\quad \Isf-2\beta\Tsf \text{ is nonexpansive.} \nonumber
\end{equation}
\end{proposition}

\begin{proof}
Let $\Rsf\defn\Isf-2\beta\Tsf$, then $\Tsf=1/(2\beta)(\Isf-\Rsf)$.
First suppose that $\Tsf$ is $\beta$-cocoercive. Let $\hbm \defn \xbm - \ybm$ for any $\xbm, \ybm \in \R^n$. We then have
\begin{equation*}
\beta\|\Tsf(\xbm)-\Tsf(\ybm)\|^2 \leq (\Tsf(\xbm)-\Tsf(\ybm))^\Tsf\hbm = \frac{1}{2\beta}\|\hbm\|^2 - \frac{1}{2\beta}(\Rsf(\xbm)-\Rsf(\ybm))^\Tsf\hbm.
\end{equation*}
We also have that
\begin{equation*}
\beta\|\Tsf(\xbm)-\Tsf(\ybm)\|^2 = \frac{1}{4\beta}\|\hbm\|^2 - \frac{1}{2\beta}(\Rsf(\xbm)-\Rsf(\ybm))^\Tsf\hbm + \frac{1}{4\beta}\|\Rsf(\xbm)-\Rsf(\ybm)\|^2.
\end{equation*}
By combining these two and simplifying the expression
\begin{equation*}
\|\Rsf(\xbm)-\Rsf(\ybm)\| \leq \|\hbm\|.
\end{equation*}
The converse can be proved by following this logic in reverse.

\end{proof}

The following characterization is also convenient.
\begin{proposition}
\label{Prop:AveragedEquiv}
For nonexpansive operators $\Tsf_1$ and $\Tsf_2$ with a constant $\alpha \in (0, 1)$, then the convex combination of the two operators $(1-\alpha)\Tsf_1 + \alpha\Tsf_2$ is nonexpansive.
\end{proposition}

\begin{proof}
Let $\Tsf\defn(1-\alpha)\Tsf_1 + \alpha\Tsf_2$. For any $\xbm,\ybm\in\R^n$, we can write
$$\|\Tsf(\xbm)-\Tsf(\ybm)\|\leq(1-\alpha)\|\Tsf_1(\xbm)-\Tsf_1(\ybm)\|+\alpha\|\Tsf_2(\xbm)-\Tsf_2(\ybm)\|\leq\|\xbm-\ybm\|$$
\end{proof}

\section{Additional Technical Details}
\label{Sup:Experiments}

This section presents several technical details that were omitted from the main paper for space.
Section~\ref{Sec:ArchitectureTraining} presents the architecture and training of our DnCNN prior.
Section~\ref{Sec:ExtraValidations} provides extra details and validations that compliment the experiments in Section~\ref{Sec:Experiments} of the main paper.

\begin{figure}[t]
\centering\includegraphics[width=0.99\linewidth]{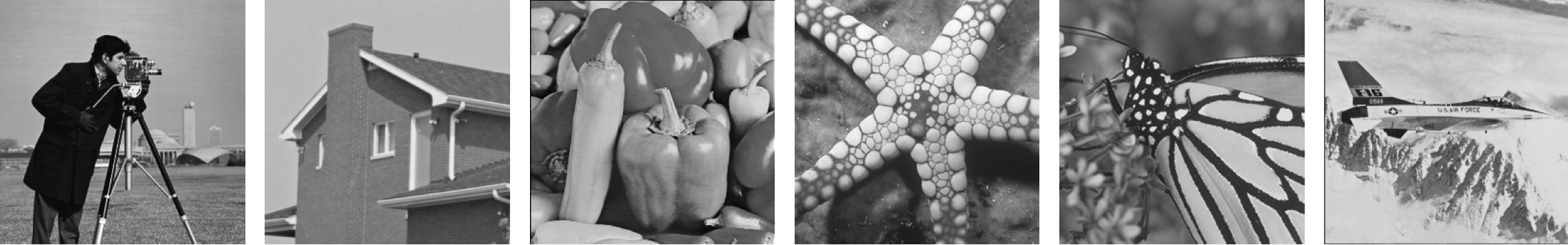}
\caption{Six test images used in the experiments on CS. From the left to right, there are \emph{cameraman}, \emph{house}, \emph{pepper}, \emph{starfish}, \emph{butterfly}, and \emph{jet}.}
\label{Fig:TestImages}
\end{figure}

\subsection{Architecture and Training of the DnCNN Prior}
\label{Sec:ArchitectureTraining}

Our denoiser follows the standard architecture of DnCNN~\cite{Zhang.etal2017}. Fig.~\ref{Fig:Architecture} visualizes the architectural details of the DnCNN prior used in our experiments.
Similar priors are extensively used in various PnP and RED algorithms~\cite{Zhang.etal2017a,Ryu.etal2019,Sun.etal2019c}.
In total, the network contains $7$ layers, of which the first $6$ layers consist of a convolutional layer and a rectified linear unit (ReLU), while the last layer contains only  a convolution operation.
A skip connection from the input to the output is used to enforce the residual network $\Rsf_\sigma$ to predict the noise residual.
The output images of the first $6$ layers have $64$ feature maps, while that of the last layer is a single-channel image.
We set all convolutional kernels to be $3\times3$ with stride $1$, which indicates that intermediate images have the same spatial size as the input image.
We generated $44700$ training examples by adding AWGN to $400$ images from the BSD400 dataset~\cite{Martin.etal2001} and extracting small patches of  $128 \times 128$ pixels with stride $30$.
Our DnCNN denoiser is trained to optimize the \emph{mean squared error} by using the Adam optimizer~\cite{Kingma.Ba2015}.

Different approaches have been used to constrain the Lipschitz constant (LC) of the denoising prior~\cite{Ryu.etal2019,Sun.etal2019c}.
We adopt the spectral normalization technique in~\cite{Sedghi2018} to control the LC of our DnCNN prior. In the training, we constrain the residual network $\Rsf_\sigma$ such that its LC is smaller than $2$.
Since the non-expansiveness of $\Dsf_\sigma$ implies that $\Rsf_\sigma$ has LC $\leq2$, this provides a \emph{necessary} condition for $\Dsf_\sigma$ to satisfy Assumption~\ref{As:NonexpansiveDen}~\cite{Sun.etal2019c}.

\subsection{Extra Details and Validations}
\label{Sec:ExtraValidations}

All experiments are run on the server equipped with 32 Intel(R) Xeon(R) CPU E5-2620 v4 processors of 3.2 GHz and 264 GBs of DDR memory. We trained all neural nets using NVIDIA RTX 2080 GPUs. We define the SNR (dB) used in the experiments as
\begin{equation}
\operatorname{SNR}(\hat{\xbm}, \xbm) \triangleq 20 \operatorname{log} _{10}\left(\frac{\|\xbm\|_{2}}{\|\xbm-\hat{\xbm}\|_{2}}\right)\nonumber
\end{equation}
where $\hat{\xbm}$ represents the reconstructed image and $\xbm$ denotes the ground truth.

Fig.~\ref{Fig:TestImages} shows the six test images used in the experiments of CS.
They are resized to the size of $240\times240$ pixels by using the Matlab function \texttt{imresize}.
As demonstrated in the middle figure in Fig.~\ref{Fig:ConvergenceSG}, \proposed-SG converges faster than \proposed-BG given a fixed amount of time.
This is further visualized in Fig.~\ref{Fig:VisualCS}, where each algorithm is run for roughly $700$ seconds.
Since \proposed-SG uses only one-fourth of the total measurements, the per-iteration complexity is lower than \proposed-BG, leading to the faster convergence speed.
In particular, the final SNR value obtained by \proposed-SG is roughly $2$ dB higher than \proposed-BG.
Additionally, both \proposed-BG/SG achieves significantly better results than \textsc{Sync-RED} and \textsc{Gm-RED} due to their adoption of asynchronous updates.

The test image used in the experiment of CT is selected from the dataset of human protein atlas~\cite{Williams.etal2017}.
We download $51$ images that have the size of $3000\times3000$ pixels.
We select one image for test, which is cropped to $800\times800$ pixels.
We extract $39000$ patches from the rest $50$ images to train five specific DnCNN denoisers for the removal of AWGN with $\sigma\in\{5,10,15,20,25\}$.
We report the result that has the highest SNR values.
The Radon matrix used in the experiments corresponds to $180$ angles with $1131$ detectors.
We synthesize the measurements by multiplying the Radom matrix with the vectorized image and add AWGN corresponding to $70$ dB input SNR.
In all tests, \proposed-SG randomly uses the measurements of $60$ angles at each iteration, while \proposed-BG uses the entire measurement set.
Fig.~\ref{Fig:CompareCT} provides a complete comparison between \proposed-BG/SG, \textsc{Sync-RED}, and \textsc{Gm-RED}.
As reference, we also include the proximal gradient method with total variation regularizer (PGM-TV).
The visual result of each method is obtained by running the algorithm with a time budget of $1$ hour.
Specifically, the per-iteration time cost of \proposed-BG/SG, \textsc{Sync-RED}, \textsc{Gm-RED}, \textsc{PGM-TV} are $5.23$, $3.21$, and $13.13$, $19.19$, and $44.74$ seconds, respectively.
The results clearly demonstrate that \proposed~are indeed effective and efficient for a realistic, nontrivial imaging task on a large-scale image.

\begin{figure}[t]
\centering\includegraphics[width=0.99\linewidth]{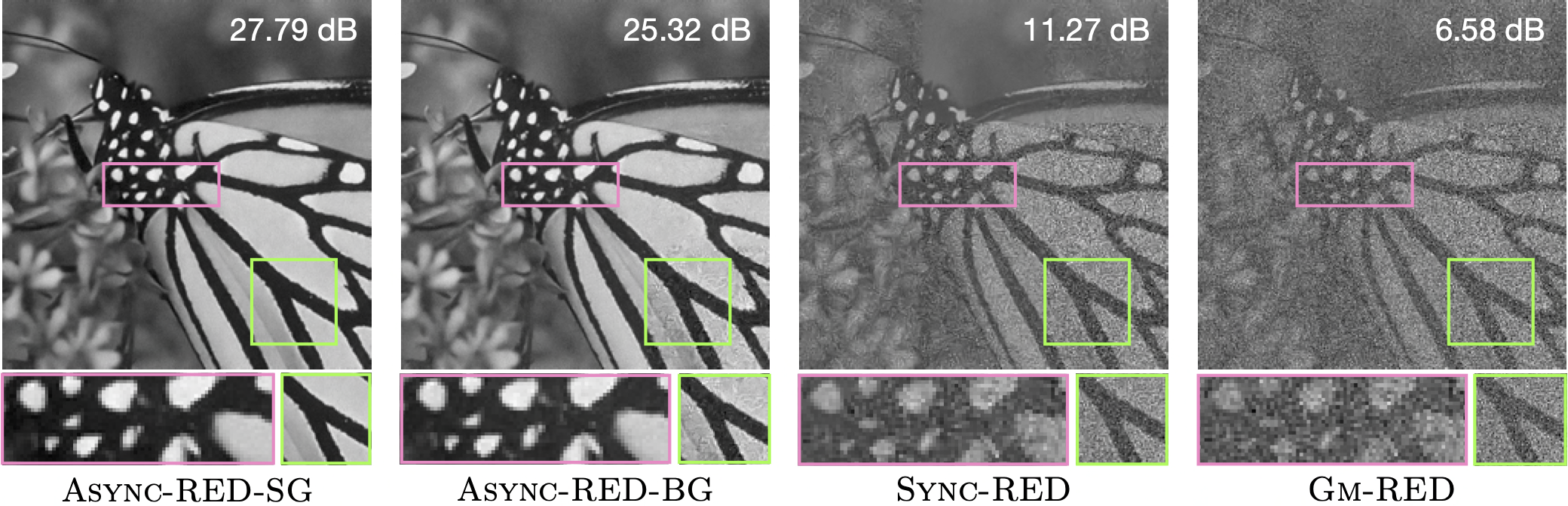}
\caption{Visualization of the recovered images from the compressed measurements by \proposed-BG/SG, \textsc{Sync-RED}, and \textsc{Gm-RED}. Each algorithm is run with a time budget of $700$ seconds.}
\label{Fig:VisualCS}
\end{figure}

\begin{figure}[t]
\centering\includegraphics[width=0.99\linewidth]{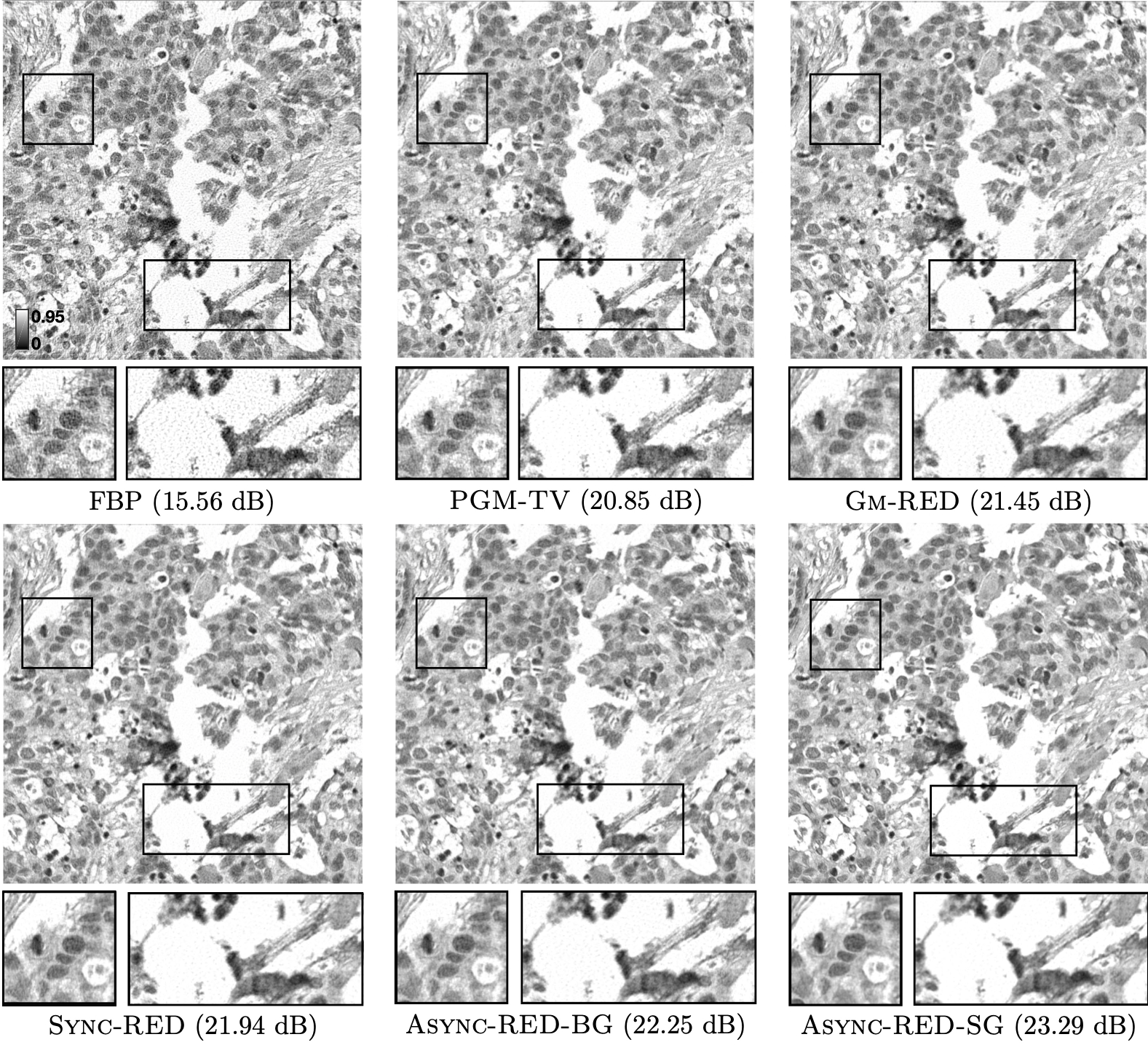}
\caption{Visualization of the reconstructed CT images by \textsc{PGM-TV}, \textsc{Gm-RED}, \textsc{Sync-RED}, and \proposed-BG/SG. Each algorithm is run with a time budget of $1$ hour. The colormap is adjusted for the best visual quality.}
\label{Fig:CompareCT}
\end{figure}


\begin{thebibliography}{10}
\providecommand{\url}[1]{#1}
\csname url@samestyle\endcsname
\providecommand{\newblock}{\relax}
\providecommand{\bibinfo}[2]{#2}
\providecommand{\BIBentrySTDinterwordspacing}{\spaceskip=0pt\relax}
\providecommand{\BIBentryALTinterwordstretchfactor}{4}
\providecommand{\BIBentryALTinterwordspacing}{\spaceskip=\fontdimen2\font plus
\BIBentryALTinterwordstretchfactor\fontdimen3\font minus
  \fontdimen4\font\relax}
\providecommand{\BIBforeignlanguage}[2]{{%
\expandafter\ifx\csname l@#1\endcsname\relax
\typeout{** WARNING: IEEEtran.bst: No hyphenation pattern has been}%
\typeout{** loaded for the language `#1'. Using the pattern for}%
\typeout{** the default language instead.}%
\else
\language=\csname l@#1\endcsname
\fi
#2}}
\providecommand{\BIBdecl}{\relax}
\BIBdecl

\bibitem{Rudin.etal1992}
L.~I. Rudin, S.~Osher, and E.~Fatemi, ``Nonlinear total variation based noise
  removal algorithms,'' \emph{Physica D}, vol.~60, no. 1--4, pp. 259--268,
  November 1992.

\bibitem{Figueiredo.Nowak2001}
M.~A.~T. Figueiredo and R.~D. Nowak, ``Wavelet-based image estimation: An
  empirical {B}ayes approach using {J}effreys' noninformative prior,''
  \emph{IEEE Trans. Image Process.}, vol.~10, no.~9, pp. 1322--1331, September
  2001.

\bibitem{Figueiredo.Nowak2003}
------, ``An {EM} algorithm for wavelet-based image restoration,'' \emph{IEEE
  Trans. Image Process.}, vol.~12, no.~8, pp. 906--916, August 2003.

\bibitem{Hu.etal2012}
Y.~Hu, S.~G. Lingala, and M.~Jacob, ``A fast majorize-minimize algorithm for
  the recovery of sparse and low-rank matrices,'' \emph{IEEE Trans. Image
  Process.}, vol.~21, no.~2, pp. 742--753, February 2012.

\bibitem{Elad.Aharon2006}
M.~Elad and M.~Aharon, ``Image denoising via sparse and redundant
  representations over learned dictionaries,'' \emph{IEEE Trans. Image
  Process.}, vol.~15, no.~12, pp. 3736--3745, December 2006.

\bibitem{Venkatakrishnan.etal2013}
S.~V. Venkatakrishnan, C.~A. Bouman, and B.~Wohlberg, ``Plug-and-play priors
  for model based reconstruction,'' in \emph{Proc. IEEE Global Conf. Signal
  Process. and Inf. Process. ({GlobalSIP})}, Austin, TX, USA, December 3-5,
  2013, pp. 945--948.

\bibitem{Sreehari.etal2016}
S.~Sreehari, S.~V. Venkatakrishnan, B.~Wohlberg, G.~T. Buzzard, L.~F. Drummy,
  J.~P. Simmons, and C.~A. Bouman, ``Plug-and-play priors for bright field
  electron tomography and sparse interpolation,'' \emph{IEEE Trans. Comput.
  Imaging}, vol.~2, no.~4, pp. 408--423, December 2016.

\bibitem{Romano.etal2017}
Y.~Romano, M.~Elad, and P.~Milanfar, ``The little engine that could:
  {R}egularization by denoising ({RED}),'' \emph{SIAM J. Imaging Sci.},
  vol.~10, no.~4, pp. 1804--1844, 2017.

\bibitem{Mataev.etal2019}
G.~Mataev, P.~Milanfar, and M.~Elad, ``Deepred: Deep image prior powered by
  red,'' in \emph{Proceedings of the IEEE/CVF International Conference on
  Computer Vision (ICCV) Workshops}, Oct 2019.

\bibitem{Metzler.etal2018}
C.~Metzler, P.~Schniter, A.~Veeraraghavan, and R.~Baraniuk, ``pr{D}eep: Robust
  phase retrieval with a flexible deep network,'' in \emph{Proc. 35th Int.
  Conf. Machine Learning (ICML)}, Stockholmsm{\"a}ssan, Stockholm Sweden,
  10--15 Jul 2018, pp. 3501--3510.

\bibitem{Wu.etal2020}
Z.~{Wu}, Y.~{Sun}, A.~{Matlock}, J.~{Liu}, L.~{Tian}, and U.~S. {Kamilov},
  ``Simba: Scalable inversion in optical tomography using deep denoising
  priors,'' \emph{IEEE Journal of Selected Topics in Signal Processing}, pp.
  1--1, 2020.

\bibitem{Reehorst.Schniter2019}
E.~T. Reehorst and P.~Schniter, ``Regularization by denoising: {C}larifications
  and new interpretations,'' \emph{IEEE Trans. Comput. Imag.}, vol.~5, no.~1,
  pp. 52--67, Mar. 2019.

\bibitem{Cohen.etal2020}
R.~Cohen, M.~Elad, and P.~Milanfar, ``Regularization by denoising via
  fixed-point projection {(RED-PRO)},'' \emph{arXiv:2008.00226 [eess.IV]},
  2020.

\bibitem{Sun.etal2019c}
Y.~Sun, J.~Liu, and U.~S. Kamilov, ``Block coordinate regularization by
  denoising,'' in \emph{Advances in Neural Information Processing Systems 32},
  Vancouver, BC, Canada, Dec. 2019, pp. 380--390.

\bibitem{Lian.etal2015}
X.~Lian, Y.~Huang, Y.~Li, and J.~Liu, ``Asynchronous parallel stochastic
  gradient for nonconvex optimization,'' in \emph{Advances in Neural
  Information Processing Systems 28}, Montreal, QC, Canada, 2015, pp.
  2737--2745.

\bibitem{Sun.etal2017}
T.~Sun, R.~Hannah, and W.~Yin, ``Asynchronous coordinate descent under more
  realistic assumption,'' in \emph{Advances in Neural Information Processing
  Systems 30}, Long Beach, California, USA, Dec. 2017, pp. 6183--6191.

\bibitem{Zhang.etal2017}
K.~Zhang, W.~Zuo, Y.~Chen, D.~Meng, and L.~Zhang, ``Beyond a {G}aussian
  denoiser: {R}esidual learning of deep {CNN} for image denoising,'' \emph{IEEE
  Trans. Image Process.}, vol.~26, no.~7, pp. 3142--3155, July 2017.

\bibitem{Bigdeli.etal2017}
S.~A. Bigdeli, M.~Jin, P.~Favaro, and M.~Zwicker, ``Deep mean-shift priors for
  image restoration,'' in \emph{Proc. Advances in Neural Information Processing
  Systems 30}, Long Beach, CA, USA, Dec 2017.

\bibitem{Zhang.etal2017a}
K.~Zhang, W.~Zuo, S.~Gu, and L.~Zhang, ``Learning deep {CNN} denoiser prior for
  image restoration,'' in \emph{Proc. {IEEE} Conf. Computer Vision and Pattern
  Recognition ({CVPR})}, Honolulu, USA, July 21-26, 2017, pp. 3929--3938.

\bibitem{Sun.etal2018b}
Y.~{Sun}, S.~{Xu}, Y.~{Li}, L.~{Tian}, B.~{Wohlberg}, and U.~S. {Kamilov},
  ``Regularized fourier ptychography using an online plug-and-play algorithm,''
  in \emph{Proc. {IEEE} Int. Conf. Acoustics, Speech and Signal Process.
  ({ICASSP})}, Brighton, UK, May 12-17, 2019, pp. 7665--7669.

\bibitem{Zhang.etal2019}
K.~Zhang, W.~Zuo, and L.~Zhang, ``Deep plug-and-play super-resolution for
  arbitrary blur kernels,'' in \emph{Proc. {IEEE} Conf. Computer Vision and
  Pattern Recognition ({CVPR})}, Long Beach, CA, USA, Jun. 2019, pp.
  1671--1681.

\bibitem{Ahmad.etal2020}
R.~{Ahmad}, C.~A. {Bouman}, G.~T. {Buzzard}, S.~{Chan}, S.~{Liu}, E.~T.
  {Reehorst}, and P.~{Schniter}, ``Plug-and-play methods for magnetic resonance
  imaging: Using denoisers for image recovery,'' \emph{IEEE Signal Processing
  Magazine}, vol.~37, no.~1, pp. 105--116, 2020.

\bibitem{Wei.etal2020}
K.~Wei, A.~Aviles-Rivero, J.~Liang, Y.~Fu, C.-B. Schnlieb, and H.~Huang,
  ``Tuning-free plug-and-play proximal algorithm for inverse imaging
  problems,'' in \emph{Proc. 37th Int. Conf. Machine Learning ({ICML})}, 2020.

\bibitem{Chan.etal2016}
S.~H. Chan, X.~Wang, and O.~A. Elgendy, ``Plug-and-play {ADMM} for image
  restoration: Fixed-point convergence and applications,'' \emph{IEEE Trans.
  Comp. Imag.}, vol.~3, no.~1, pp. 84--98, March 2017.

\bibitem{Meinhardt.etal2017}
T.~Meinhardt, M.~Moeller, C.~Hazirbas, and D.~Cremers, ``Learning proximal
  operators: {U}sing denoising networks for regularizing inverse imaging
  problems,'' in \emph{Proc. IEEE Int. Conf. Comp. Vis. (ICCV)}, Venice, Italy,
  October 22-29, 2017, pp. 1799--1808.

\bibitem{Buzzard.etal2017}
G.~T. Buzzard, S.~H. Chan, S.~Sreehari, and C.~A. Bouman, ``Plug-and-play
  unplugged: {O}ptimization free reconstruction using consensus equilibrium,''
  \emph{SIAM J. Imaging Sci.}, vol.~11, no.~3, pp. 2001--2020, September 2018.

\bibitem{Sun.etal2018a}
Y.~Sun, B.~Wohlberg, and U.~S. Kamilov, ``An online plug-and-play algorithm for
  regularized image reconstruction,'' \emph{IEEE Trans. Comput. Imaging}, 2019.

\bibitem{Tirer.Giryes2019}
T.~Tirer and R.~Giryes, ``Image restoration by iterative denoising and backward
  projections,'' \emph{IEEE Trans. Image Process.}, vol.~28, no.~3, pp.
  1220--1234, Mar. 2019.

\bibitem{Teodoro.etal2019}
A.~M. Teodoro, J.~M. Bioucas-Dias, and M.~Figueiredo, ``A convergent image
  fusion algorithm using scene-adapted {G}aussian-mixture-based denoising,''
  \emph{IEEE Trans. Image Process.}, vol.~28, no.~1, pp. 451--463, Jan. 2019.

\bibitem{Ryu.etal2019}
E.~K. Ryu, J.~Liu, S.~Wang, X.~Chen, Z.~Wang, and W.~Yin, ``Plug-and-play
  methods provably converge with properly trained denoisers,'' in \emph{Proc.
  36th Int. Conf. Machine Learning (ICML)}, 2019, pp. 5546--5557.

\bibitem{Xu.etal2020}
X.~{Xu}, Y.~{Sun}, J.~{Liu}, B.~{Wohlberg}, and U.~S. {Kamilov}, ``Provable
  convergence of plug-and-play priors with mmse denoisers,'' \emph{IEEE Signal
  Processing Letters}, vol.~27, pp. 1280--1284, 2020.

\bibitem{Metzler.etal2016a}
C.~A. Metzler, A.~Maleki, and R.~Baraniuk, ``{BM3D}-{PRGAMP}: {C}ompressive
  phase retrieval based on {BM3D} denoising,'' in \emph{Proc. {IEEE} Int. Conf.
  Image Proc. (ICIP)}, Phoenix, AZ, USA, September 25-28, 2016, pp. 2504--2508.

\bibitem{Metzler.etal2016}
C.~A. Metzler, A.~Maleki, and R.~G. Baraniuk, ``From denoising to compressed
  sensing,'' \emph{IEEE Trans. Inf. Theory}, vol.~62, no.~9, pp. 5117--5144,
  September 2016.

\bibitem{Fletcher.etal2018}
A.~K. Fletcher, P.~Pandit, S.~Rangan, S.~Sarkar, and P.~Schniter, ``Plug-in
  estimation in high-dimensional linear inverse problems: A rigorous
  analysis,'' in \emph{Advances in Neural Information Processing Systems 31},
  Montreal, QC, Canada, Dec. 2018, pp. 7451--7460.

\bibitem{Rangan.etal2014}
S.~Rangan, P.~Schniter, and A.~Fletcher, ``On the convergence of approximate
  message passing with arbitrary matrices,'' in \emph{Proc. IEEE Int. Symp.
  Information Theory}, Honolulu, HI, USA, June 29-July 4, 2014, pp. 236--240.

\bibitem{Rangan.etal2015}
S.~Rangan, A.~K. Fletcher, P.~Schniter, and U.~S. Kamilov, ``Inference for
  generalized linear models via alternating directions and {Bethe} free energy
  minimization,'' in \emph{Proc. IEEE Int. Symp. Information Theory}, Hong
  Kong, June 14-19, 2015, pp. 1640--1644.

\bibitem{Liu.etal2013a}
J.~Liu, S.~J. Wright, C.~R\'{e}, V.~Bittorf, and S.~Sridhar, ``An asynchronous
  parallel stochastic coordinate descent algorithm,'' \emph{J. Mach. Learn.
  Res.}, vol.~16, no.~1, pp. 285--322, Jan. 2015.

\bibitem{Peng.etal2015}
Z.~Peng, Y.~Xu, M.~Yan, and W.~Yin, ``Arock: An algorithmic framework for
  asynchronous parallel coordinate updates,'' \emph{SIAM Journal on Scientific
  Computing}, vol.~38, no.~5, pp. A2851--A2879, 2016.

\bibitem{Hannah.etal2018}
R.~Hannah and W.~Yin, ``On unbounded delays in asynchronous parallel
  fixed-point algorithms,'' \emph{Journal of Scientific Computing}, vol.~76,
  no.~1, pp. 299--326, Jul 2018.

\bibitem{Hannah.etal2018b}
R.~Hannah, F.~Feng, and W.~Yin, ``A2{BCD}: Asynchronous acceleration with
  optimal complexity,'' in \emph{International Conference on Learning
  Representations}, 2019.

\bibitem{Recht.etal2011}
B.~Recht, C.~Re, S.~Wright, and F.~Niu, ``Hogwild: A lock-free approach to
  parallelizing stochastic gradient descent,'' in \emph{Advances in Neural
  Information Processing Systems 24}, Granada, Spain, Dec 2011, pp. 693--701.

\bibitem{Liu.etal2018b}
T.~Liu, S.~Li, J.~Shi, E.~Zhou, and T.~Zhao, ``Towards understanding
  acceleration tradeoff between momentum and asynchrony in nonconvex stochastic
  optimization,'' in \emph{Advances in Neural Information Processing Systems
  31}, Montreal, QC, Canada, Dec 2018, pp. 3682--3692.

\bibitem{Zhou.etal2018}
Z.~Zhou, P.~Mertikopoulos, N.~Bambos, P.~Glynn, Y.~Ye, L.~Li, and L.~F.,
  ``Distributed asynchronous optimization with unbounded delays: How slow can
  you go?'' in \emph{Proc. 35th Int. Conf. Machine Learning ({ICML})},
  Stockholmsm{\"a}ssan, Stockholm Sweden, 10--15 Jul 2018, pp. 5970--5979.

\bibitem{Lian.etal2018}
X.~Lian, W.~Zhang, C.~Zhang, and J.~Liu, ``Asynchronous decentralized parallel
  stochastic gradient descent,'' in \emph{Proc. 35th Int. Conf. Machine
  Learning ({ICML})}, Stockholmsm{\"a}ssan, Stockholm Sweden, 10--15 Jul 2018,
  pp. 3043--3052.

\bibitem{Liu.etal2015b}
J.~Liu and S.~J. Wright, ``Asynchronous stochastic coordinate descent:
  Parallelism and convergence properties,'' \emph{SIAM Journal on
  Optimization}, vol.~25, no.~1, pp. 351--376, 2015.

\bibitem{Wright2015}
S.~J. Wright, ``Coordinate descent algorithms,'' \emph{Math. Program.}, vol.
  151, no.~1, pp. 3--34, Jun. 2015.

\bibitem{Peng.etal2019}
Z.~Peng, Y.~Xu, M.~Yan, and W.~Yin, ``On the convergence of asynchronous
  parallel iteration with unbounded delays,'' \emph{Journal of the Operations
  Research Society of China}, vol.~7, no.~1, pp. 5--42, 2019.

\bibitem{Nesterov2012}
Y.~Nesterov, ``Efficiency of coordinate descent methods on huge-scale
  optimization problems,'' \emph{SIAM J. Optim.}, vol.~22, no.~2, pp. 341--362,
  2012.

\bibitem{Beck.Tetruashvili2013}
A.~Beck and L.~Tetruashvili, ``On the convergence of block coordinate descent
  type methods,'' \emph{SIAM J. Optim.}, vol.~23, no.~4, pp. 2037--2060, Oct.
  2013.

\bibitem{Ghadimi.Lan2016}
S.~Ghadimi and G.~Lan, ``Accelerated gradient methods for nonconvex nonlinear
  and stochastic programming,'' \emph{Math. Program. Ser. A}, vol. 156, no.~1,
  pp. 59--99, March 2016.

\bibitem{Miyato.etal2018}
T.~Miyato, T.~Kataoka, M.~Koyama, and Y.~Yoshida, ``Spectral normalization for
  generative adversarial networks,'' in \emph{International Conference on
  Learning Representations}, 2018.

\bibitem{Sedghi2018}
H.~Sedghi, V.~Gupta, and P.~M. Long, ``The singular values of convolutional
  layers,'' in \emph{International Conference on Learning Representations},
  2019.

\bibitem{Anil.etal2019}
C.~Anil, J.~Lucas, and R.~Grosse, ``Sorting out {L}ipschitz function
  approximation,'' in \emph{Proc. 36th Int. Conf. Machine Learning ({ICML})},
  Long Beach, California, USA, 09--15 Jun 2019, pp. 291--301.

\bibitem{Nemirovski.etal2009}
A.~Nemirovski, A.~Juditsky, G.~Lan, and A.~Shapiro, ``Robust stochastic
  approximation approach to stochastic programming,'' \emph{SIAM J. Optim.},
  vol.~19, no.~4, pp. 1574--1609, 2009.

\bibitem{Dekel.etal2010}
O.~Dekel, R.~Gilad{-}Bachrach, O.~Shamir, and L.~Xiao, ``Optimal distributed
  online prediction using mini-batches,'' \emph{Journal of Machine Learning
  Research}, vol.~13, no.~1, pp. 165--202, 2012.

\bibitem{Ryu.Boyd2016}
E.~K. Ryu and S.~Boyd, ``A primer on monotone operator methods,'' \emph{Appl.
  Comput. Math.}, vol.~15, no.~1, pp. 3--43, 2016.

\bibitem{Rockafellar.Wets1998}
R.~T. Rockafellar and R.~Wets, \emph{Variational Analysis}.\hskip 1em plus
  0.5em minus 0.4em\relax Springer, 1998.

\bibitem{Boyd.Vandenberghe2004}
S.~Boyd and L.~Vandenberghe, \emph{Convex Optimization}.\hskip 1em plus 0.5em
  minus 0.4em\relax Cambridge Univ. Press, 2004.

\bibitem{Nesterov2004}
Y.~Nesterov, \emph{Introductory Lectures on Convex Optimization: A Basic
  Course}.\hskip 1em plus 0.5em minus 0.4em\relax Kluwer Academic Publishers,
  2004.

\bibitem{Bauschke.Combettes2017}
H.~H. Bauschke and P.~L. Combettes, \emph{Convex Analysis and Monotone Operator
  Theory in Hilbert Spaces}, 2nd~ed.\hskip 1em plus 0.5em minus 0.4em\relax
  Springer, 2017.

\bibitem{Martin.etal2001}
D.~Martin, C.~Fowlkes, D.~Tal, and J.~Malik, ``A database of human segmented
  natural images and its application to evaluating segmentation algorithms and
  measuring ecological statistics,'' in \emph{Proc. {IEEE} Int. Conf. Comp.
  Vis. ({ICCV})}, Vancouver, Canada, July 7-14, 2001, pp. 416--423.

\bibitem{Kingma.Ba2015}
D.~Kingma and J.~Ba, ``Adam: {A} method for stochastic optimization,'' in
  \emph{International Conference on Learning Representations (ICLR)}, 2015,
  arXiv:1412.6980 [cs.LG].

\bibitem{Williams.etal2017}
E.~Williams, J.~Moore, S.~W. Li, G.~Rustici, A.~Tarkowska, A.~Chessel, S.~Leo,
  B.~Antal, R.~K. Ferguson, U.~Sarkans \emph{et~al.}, ``Image data resource: a
  bioimage data integration and publication platform,'' \emph{Nature methods},
  vol.~14, no.~8, pp. 775--781, 2017.

\end{thebibliography}
\end{document}